\documentclass[11pt,a4paper]{article}
\usepackage[a4paper, a4paper, left=3cm, right=3cm, top=3cm, bottom=3cm]{geometry}
\usepackage{authblk}
\usepackage{float}
\usepackage[T1]{fontenc}
\usepackage[utf8]{inputenc}
\usepackage{amsmath,amsthm,amsfonts,amssymb}
\usepackage{nccmath}
\usepackage{graphicx}
\usepackage{enumerate}
\usepackage{cite}
\usepackage{caption}
\usepackage[labelformat=simple]{subcaption}
\usepackage{bm}
\newtheorem{proposition}{\bf Proposition}
\newtheorem{lemma}{\bf Lemma}
\newtheorem{remark}{\bf Remark}

\title{Closed form perturbation theory in the restricted three-body problem without relegation}\date{}
\author[1]{Irene Cavallari,}
\author[2]{Christos Efthymiopoulos}
\affil[1]{Dipartimento di Matematica, Universit\`a di Pisa}
\affil[2]{Dipartimento di Matematica Tullio Levi Civita, Universit\`a degli studi di Padova}

\def \br{\textbf{r}}
\def\brp{\textbf{r}_P}
\def \bp{\textbf{p}}

\def \aV{a^*}

\def \etaV{\eta^*}
\def \etaVP{\eta_+}
\def \uV{u^*}
\begin{document}
	
	\maketitle	
		
\begin{abstract}
We propose a closed-form normalization method suitable for the study of the secular dynamics of small bodies in heliocentric orbits perturbed by the tidal potential of a planet with orbit external to the orbit of the small body. The method makes no use of relegation, thus, circumventing all convergence issues related to that technique. The method is based on a convenient use of a book-keeping parameter keeping simultaneously track of all the small quantities in the problem. The book-keeping affects both the Lie series and the Poisson structure employed in successive perturbative steps. In particular, it affects the definition of the normal form remainder at every normalization step. We show the results obtained by assuming Jupiter as perturbing planet and we discuss the validity and limits of the method. 
\end{abstract}

\section{Introduction}
The present paper is about the development of a method to compute a secular normal form in the framework of the restricted three-body problem (R3BP). We are interested in the heliocentric dynamics of a massless particle perturbed by the tidal potential of an external planet, i.e. a planet whose orbit is exterior to the particle's one. The objective is to define a transformation leading to a Hamiltonian function suitable to study the particle's secular dynamics, i.e. we search for a normal form not depending on the fast angles characterizing the problem. Using modified Delaunay variables, the latter are the mean longitudes of the particle and of the planet. 

\paragraph{}
The Hamiltonian for the problem of interest is equal to the sum of two components,  a leading term $Z_0$ not depending on the fast angles and a disturbing function $R$:
\[
\mathcal{H}=Z_0+ R.
\] 
The leading term is
\[
Z_0 = -\frac{(\mathcal{GM})^2}{2\Lambda}+n_PI_P,
\]
where $n_P$ is the planet's mean motion and $I_P$ is a dummy action variable canonically conjugated to the planet's mean longitude $\lambda_P$ such that
\[
\frac{d \lambda_P}{dt}=\frac{\partial \mathcal{H}}{\partial I_P}.
\]
Moreover, we have
\[
R = \mu \tilde{R}, \qquad \tilde{R}=\tilde{R} (a(\Lambda),e(\Lambda,\Gamma),i(\Lambda,\Gamma,\Theta),\Omega(\theta),\omega(\gamma,\theta),u(\Lambda,\Gamma,\lambda),f_P(\lambda_P); e_P, a_P)
\] 
with $a_P$, $e_P$, $f_P$ the semi-major axis, eccentricity and true anomaly of the planet and $a$,$e$,$i$,$\Omega$,$\omega$,$u$ the orbital elements of the particle ($u$ is the eccentric anomaly); $(\Lambda,\Gamma,\Theta,\lambda,\gamma,\theta)$ are modified Delaunay variables; $\mathcal{M}$ is the mass of the Sun, $\mathcal{G}$ is the Newton's gravity parameter and $\mu=\mathcal{G}m_P$, with $m_P$ the mass of the planet.

As typical in perturbation theory, in the sequel we will write  $\mathcal{H}$ in the form of a formal series:
\begin{equation}
\mathcal{H} = Z_0 + \sum_{s=1}^{+\infty}\epsilon^sR_s
\label{Ham0Class}
\end{equation}
where $\epsilon$ is a formal parameter, called \textit{book-keeping} parameter, used to assess the size of each perturbing term composing $R$. The normal form has to be computed iteratively, through, for example, a composition of Lie transformations \cite{Deprit1969}. A Lie transformation is a canonical transformation between two sets of canonical variables, $z_1$ and $z_2$, given by
\[
z_1 = \exp(\mathcal{L_{\chi}})z_2
\]
with the operator $\exp(\mathcal{L_{\chi}})$ defined as
\[
\exp(\mathcal{L_{\chi}})=\sum_{j=0}^{\infty}\frac{1}{j!}\mathcal{L_{\chi_j}^j},
\]
where $\mathcal{L_{\chi}}=\{\cdot,\chi\}$ is the Poisson bracket operator and 
\[
\mathcal{L_{\chi}^j}f = \underbrace{\{\dots\{\{f,\chi\},\chi\}\dots,\chi\}}_{j \hspace{1mm}\mbox{times}},
\]
see \cite{Christos}. The function $\chi$ is called a Lie generating function.
Performing $r$ steps of the normalization process, the transformation between the original and the final variables $z^{(0)}$, $z^{(r)}$
\[
z^{(0)}=(\Lambda^{(0)}, \Gamma^{(0)},\Theta^{(0)},I_P^{(0)},\lambda^{(0)},\gamma^{(0)},\theta^{(0)},\lambda_P^{(0)}),\]
\[z^{(r)}=(\Lambda^{(r)}, \Gamma^{(r)},\Theta^{(r)},I_P^{(r)},\lambda^{(r)},\gamma^{(r)},\theta^{(r)},\lambda_P^{(r)})\]
 is
\[
z^{(0)}=\exp(\mathcal{L}_{\chi_r})\exp(\mathcal{L}_{\chi_{r-1}})\dots \exp(\mathcal{L}_{\chi_1})z^{(r)}.
\]  
At the $j$-th iteration,  we perform a Lie transformation leading to a new Hamiltonian $\mathcal{H}^{(j)}$ given by 
\[
\mathcal{H}^{(j)} = \exp(\mathcal{L}_{\chi_j})\mathcal{H}^{(j-1)}
\]
with
\[
\mathcal{H}^{(j-1)} = Z_0+\sum_{s=1}^{j-1}\epsilon^sZ_s + \sum_{s=j}^{+\infty}\epsilon^sR_s^{(j)};
\]
where $Z_i$, $i=1\dots j-1$ are normal form terms arising at previous normalization steps. We have 
\[\mathcal{H}^{(0)}=\mathcal{H}, \qquad R_s^{(0)}=R_s, \forall s\in[1,+\infty).\]

Following the above iterative procedure, the Hamiltonian $\mathcal{H}^{(r)}$  after $r$ normalization steps is given by the sum between the normal form 
\[
Z = \sum_{s=0}^{r}Z_s,
\]
and a remainder with a (hopefully) small size. Each generating function $\chi_j$, $j=1\dots r$, is determined by solving the so-called homological equation
\begin{equation}
\{Z_0,\chi_j \}+\epsilon^j R_j^{(j-1)}=\epsilon^j Z_j, 
\label{homEq0}
\end{equation}
with 
\[
\{Z_0,\chi_j\}=-\Big(n\frac{\partial}{\partial \lambda}+n_P\frac{\partial}{\partial \lambda_P}\Big)\chi_j, \qquad n = -\frac{(\mathcal{GM})^2}{\Lambda^3}.
\]
 
In our problem, solving \eqref{homEq0} can be complex, since the disturbing function $R$ is not directly available as a sum of trigonometric terms over the problem's canonical angles. In fact, $R$ depends on the mean longitudes of the planet and of the particle through geometric angles, i.e. the true anomalies or the eccentric anomalies. This implies that we have to solve Kepler's equation in series form in order to obtain the required trigonometric expansions in the angles $\lambda$, $\lambda_P$. Typically, to overcome this difficulty the original Hamiltonian is approximated by means of a Taylor expansion in some small parameter truncated at an adequate order to make explicit the dependence on the fast angles (see \cite{Tisserand1889,Brouwer1961,Kaula1966,Vinh1970,Brumberg1994}). As an example, consider an expansion over the eccentricities of both the particle and the planet. The book-keeping parameter in \eqref{Ham0Class} will depend on $e$ and $e_P$ and each term $R_s^{(j-1)}$ will have the form 
\[
R_s^{(j-1)} = \sum_{\bm{k}} R_{s,\bm{k}}^{(j-1)}\cos(k_1\lambda+k_2\lambda_P+k_3\gamma+k_4\theta)
\]
where $R_{s,\bm{k}}^{(j-1)}=R_{s,\bm{k}}^{(j-1)}(\Lambda,\Gamma,\Theta;e_P,a_P)$ are the so-called Laplace coefficients \cite{Brouwer1961}. The generating function $\chi_{j}$ satisfying \eqref{homEq0} is
\[
\chi_{j} = \sum_{\bm{k}} \frac{R_{j,\bm{k}}^{(j-1)}}{nk_1+n_Pk_2}\sin(k_1\lambda+k_2\lambda_P+k_3\gamma+k_4\theta).
\]
Some examples of application of the above method of expansion can be found in \cite{Metris1993,Wnuk1988,Lara2011}. However, an important drawback of this technique is that it can be applied only for lowly eccentric orbits. To address this issue, an alternative, introduced in \cite{Palacian1992} and formalized in \cite{Deprit2001,Segerman2000}, is the so called relegation method; it consists in neglecting one of the two components of the leading term in the homological equation, so that this can be solved in closed form. A multipole expansion of the initial disturbing function $R$ is performed so that each $R_s^{(j-1)}$ has the form
\[
R_s^{(j-1)}=\sum_{\bm{k}} \frac{f_{\bm{k}}(\Lambda,\Gamma,\Theta;e_P,a_P)}{r(\Lambda,\Gamma,\lambda)}\cos\big({k_1u(\Lambda,\Gamma,\lambda)+k_2f_P(\lambda_P)+k_3\gamma+k_4\theta}\big),
\]
where $r$ is the heliocentric radius of the particle equal to
$
r = a\big(1-e\cos{u}\big)
$.
Since the planet's trajectory is external to the particle's trajectory, we have $n_P<n$. Hence, treating $n_P$ as a small quantity (of  `book-keeping order $1$', see below), instead of the homological equation \eqref{homEq0}, we work with a homological equation involving only the particle's Keplerian terms in the kernel, namely (see \cite{Palacian2002}):
\begin{equation}
-n\frac{\partial \chi_j}{\partial \lambda} + {R}_j^{(j-1)} = \bar{{{R}}}^{(j-1)}_{j,\lambda},
\label{releghe}
\end{equation}
so that
\begin{equation}
\chi_j = \frac{1}{n}\int  ({R}_j^{(j-1)}-\bar{{{R}}}^{(j-1)}_{j,\lambda})d\lambda, \qquad d\lambda = \frac{r}{a}du,
\label{chi_releg}
\end{equation}
where $\bar{{{R}}}^{(j-1)}_{j,\lambda}$ is the average of ${R}_j^{(j-1)}$ over $\lambda$. Equation \eqref{chi_releg}, now, can be solved in  `closed-form' (see \cite{Palacian2002}), i.e. without expanding the function ${R}_j^{(j-1)}$ in the eccentricities $e$, $e_P$. However, it is easy to see that, by applying the Lie transformation, the new Hamiltonian contains terms due to the contribute $\{n_P I_P,\chi_j\}$ of the type
\[
{R}_{j,1}^{(j,0)}=\delta	f_{\bm{k}}(\Lambda,\Gamma,\Theta)\cos({k_1u+k_2f_P+k_3\gamma+k_4\theta}), \qquad \delta = \frac{k_2 n_P}{k_1 n}
\]
whose size may not be sufficiently small, i.e. comparable to the size of the next term to normalize, namely $R_{j+1}^{(j)}$. As a consequence, the iterative normalization process must be adjusted: additional iterations must be added to handle first the terms coming from $\{n_P I_P,\chi_j\}$, before dealing with $R_{j+1}^{(j)}$. This leads also to remainder terms depending on increasing powers of $\delta$, whose size decreases, provided that $\delta<1$. Generating functions $\chi_{j}^{k}$, $k=1\dots m$, satisfying  
\[
-n\frac{\partial \chi_j^k}{\partial \lambda} + {{R}}_{j,k}^{(j,k-1)} = \bar{{{R}}}^{(j,k-1)}_{j,k,\lambda},
\]
are iteratively computed to normalize the contributions ${R}_{j,k}^{(j,k-1)}$ depending on $\delta^k$. After $m$ steps, the final remainder ${{R}}_{j,m+1}^{(j,m)}$ will depend on $\delta^{m+1}$ and will have a size comparable to $R_{j+1}^{(j)}$.

While the relegation technique successfully remedies the issue of the poor convergence of series depending on powers of the orbital eccentricities (see \cite{SanJuan2004,Ceccaroni2013,Feng2015,Palacian2006} for applications), the practical applicability of the technique is severely limited by the requirement of smallness of the ratio $n_P/n$. To understand this, 
let us use the example of a two degree-of-freedom harmonic oscillator system in action-angle variables as proposed in \cite{Segerman2000}:
\[
\mathcal{H}^{(0)}= \omega_1 J_1+\omega_2 J_2 + \epsilon  \cos(k_1\phi_1+k_2\phi_2),
\]
with $J_1,J_2$ the actions and $\phi_1,\phi_2$ the angles. The homological equation to solve is
\[
\{\omega_1 J_1+\omega_2 J_2,\chi_C\}+\epsilon \cos(k_1\phi_1+k_2\phi_2)=0,
\]
where $\epsilon$ is a formal parameter to assess the size of the terms. 
The classical solution is 
\begin{equation}
\chi_C = \epsilon \frac{\sin(k_1\phi_1+k_2\phi_2)}{k_1\omega_1+k_2\omega_2}.
\label{chiCrel}
\end{equation}
However, in the case $\omega_1/\omega_2<<1$ we can apply the relegation technique: we can neglect $\omega_1J_1$ in the leading term and determine a generating function $\chi_R$ satisfying
\[
-\omega_2\frac{\partial \chi_R}{\partial \phi_2}+\epsilon \cos(k_1\phi_1+k_2\phi_2)=0,
\]
that is
\[\chi_R = \frac{\epsilon}{\omega_2 k_2}\sin(k_1\phi_1+k_2\phi_2). 
\]
Because of the term $\omega_1J_1$ in the Hamiltonian, the Lie transformation gives
\[
\mathcal{H}^{(1)}=\exp(\mathcal{L}_{\chi_R})\mathcal{H}=\omega_1 J_1+\omega_2 J_2-\frac{k_1\omega_1}{k_2\omega_2}\epsilon\cos(k_1\phi_1+k_2\phi_2).
\]

If $k_1\omega_1/k_2\omega_2<1$ the remainder size is lower than the size of the normalized term. However, it may be higher than the targeted size, so that the relegation process must be further iterated. As discussed in \cite{Segerman2000}, the iterations produce the generating function 
\begin{equation}
\chi_R = \epsilon \sin(k_1\phi_1+k_2\phi_2)\frac{1}{k_2\omega_2}\Bigg(1-\frac{k_1\omega_1}{k_2\omega_2}+\Big(\frac{k_1\omega_1}{k_2\omega_2}\Big)^2+..\Bigg).
\label{chiRrelegation}
\end{equation}
It is, now, trivial to see that the generating function $\chi_R$, obtained by relegation, corresponds to the series expansion of the usual generating function $\chi_C$, obtained without relegation, in powers of the ratio $k_1\omega_1/k_2\omega_2$. However, it is obvious that, even if $\omega_1<<\omega_2$, the method may not converge if the coefficients $k_1,k_2$ are such that $k_1\omega_1/k_2\omega_2\ge1$. We refer to \cite{Sansottera2017} for more details about the convergence of the relegation algorithm. 

\paragraph{}
Some methods alternative to relegation have been proposed in literature to solve the homological equation in closed form. In \cite{Mahajan2018}, a technique based on the method of characteristics is developed; its application is shown in \cite{Mahajan2019}. In \cite{Lara2013}, the homological equation is solved in closed form for orbits with low eccentricity by accepting a remainder of small size depending on $e$. 

\paragraph{}
All the above techniques were applied, so far, in the so-called  `satellite problem', i.e. the motion of a test body in the multiple expansion of a planet's gravitational potential (e.g. with the $J_2$ and $C_{22}$ terms). In the present paper, we examine, instead, the applicability of a closed-form normalization method in the framework of the R3BP suitable for orbits with relatively high eccentricities and not using relegation. Our method is similar in spirit to the one introduced in \cite{Lara2013} for satellite motions in the geopotential. In particular, after a multipole expansion of the initial disturbing function, we introduce a book-keeping symbol $\epsilon$ (with numerical value equal to $1$), and write the initial Hamiltonian as 
\[
\mathcal{H}^{(0)} = Z_0 + \sum_{s=s_0}^{+\infty}\epsilon^sR_s^{(0)},
\]
where we have
\begin{equation}
s_0 =\Bigg\lceil \frac{\log\big(\frac{m_P}{\mathcal{M}}\big)}{\log (e)}\Bigg\rceil.
\label{s0def}
\end{equation}
The exponent of the book-keeping parameter $\epsilon$ in each perturbing term $R_s^{(0)}$ keeps track of the order of smallness of the term, which, in turn, may depend on one of more of the following three small quantities: $e$, $e_P$ and the ratio between the planet and the Sun's masses. As in \cite{Lara2013}, to overcome the difficulty of solving the homological equation in closed form the main idea is to accept a remainder coming from the homological equation itself; at each $j$-th iteration, $j=1\dots r$, we determine a  generating function $\chi_{s_0+j-1}^{(j)}$ satisfying 
\[
\Big\{Z_0,\chi_{s_0+j-1}^{(j)}\}+\epsilon^{s_0+j-1}R_{s_0+j-1}^{(j-1)}=\epsilon^{s_0+j-1}Z_{s_0+j-1}+\mathcal{O}(\epsilon^{s_o+j})
\] 
where $Z_{s_0+j-1}$ does not depend on $\lambda$ and $\lambda_P$. The new Hamiltonian is
\[
\mathcal{H}^{(j)}=\exp(\mathcal{L}_{\chi_{s_0+j-1}^{(j)}})\mathcal{H}^{(j-1)} = Z_0 + \sum_{s=s_0}^{s_0+j-1}\epsilon^sZ_s + \sum_{s=s_0+j}^{+\infty}\epsilon^sR^{(j)}_s
\]
where $R^{(j)}_s$ contains also the remainder contributions coming from the homological equation.
\paragraph{}
The structure of the paper is as follows. The method will be detailed in Section \ref{section: method}. In Section \ref{section: outcomesPCR3BP}, we apply the method and give numerical results for the simplest case of the planar circular restricted three-body problem (PCR3BP); an analysis of the results is performed to assess the validity of the method. In the present study, we consider Jupiter as the perturbing planet and a main belt asteroid as the test particle. In Section \ref{section: outcomesR3BP} we report the outcomes obtained by applying, instead, the method to some orbits in the more general planar elliptic R3BP.

\section{Normalization Method}
\label{section: method}
In this section, we describe the formal steps required to apply the proposed closed-form normalization method. They include the preparation of the initial Hamiltonian, the choice of the book-keeping scheme, the definitions related to the used Poisson structure as well as the normalization process through the composition of Lie series. 

\subsection{Hamiltonian preparation}
Let us consider a heliocentric inertial reference frame with the $\widehat x$ axis pointing towards the planet's perihelion and the $\widehat z$ axis parallel to the planet's orbital angular momentum. The Hamiltonian of the R3BP is 
\begin{equation}
H = \frac{p^2}{2} - \frac{\mathcal{GM}}{r} + \mathcal{R}
\label{Hr3bp}
\end{equation}
where $\br$ is the particle's heliocentric position vector, $r=\vert \br \vert$, and $\bp$ is the conjugated canonical momenta vector, with $p=\vert \bp \vert$. In equation \eqref{Hr3bp}, $\mathcal{R}$ is the perturbing planet's tidal potential equal to
\begin{equation*}
\mathcal{R}=-\mu\big(\frac{1}{\sqrt{r^2+r_P^2-2\br	\cdot \brp}}-\frac{\br\cdot\brp}{r_P^3}\big)
\end{equation*}
with $\brp$ the position vector of the planet, $r_P=\vert \brp \vert$. 

\paragraph{1 - Multipolar Expansion}
We are interested in analysing the motion of small bodies orbiting the Sun for which we always have $r<r_P$. Then, the function $\mathcal{R}$ can be approximated with its truncated multipole expansion: 
\begin{equation}
\mathcal{R} \simeq \mathsf{R} = -\frac{\mu}{r_P}\sum_{j=2}^{o}\frac{r^j}{r_P^j} P_j(\cos \alpha),
\label{tidalpotentialme}
\end{equation}
where
\[
\cos \alpha = \frac{\bm{r}\cdot\bm{r}_P}{rr_P}
\]
and $P_j(\cdot)$ are Legendre polynomials. The time-dependent term $1/r_P$ is omitted in \eqref{tidalpotentialme} since it does not contribute to the particle's equations of motion.

\paragraph{2 - Extended Hamiltonian} 
The Hamiltonian \eqref{Hr3bp} can be expressed as a function of orbital elements, using the relations 
\begin{equation}
\brp = r_P\big(\cos(f_P(t)), \sin(f_P(t)), 0\big)^T, \qquad r_P = \frac{a_P\eta_P^2}{1+e_P\cos(f_P)}
\label{rPdef}
\end{equation}

\[
\br =r\Bigg(\begin{matrix} \cos(\Omega)\cos(\omega+f) - \sin(\Omega)\cos(i)\sin(\omega+f)\\
\sin(\Omega)\cos(\omega+f) + \cos(\Omega)\cos(i)\sin(\omega+f)\\
\sin(i)\sin(\omega+f)
\end{matrix}\Bigg), \qquad r =  \frac{a\eta^2}{1+e\cos(f)},
\]
 where
\[
\eta=\sqrt{1-e^2},
\]
and $f$ is the true anomaly. To avoid trigonometric functions at the denominator in $\mathsf{R}$, it turns convenient to introduce the eccentric anomaly $u$ in place of $f$ through the relations
\begin{equation}
\cos{f} =\frac{a}{r}(\cos{u}-e), \qquad \sin{f}=\frac{a}{r}\sin{u},\qquad
r=a(1-e\cos{u}).
\label{ratior}
\end{equation}
The planet orbit is assumed Keplerian, so that only the true anomaly $f_P$ varies in time. The variable $f_P$ depends on time through the orbit's mean longitude $\lambda_P$. However, the Hamiltonian can be formally extended to an autonomous one by adding a term depending on a dummy action $I_P$ conjugated to the angle $\lambda_P$. The extended Hamiltonian is 
\[
\mathsf{H}=  - \frac{\mathcal{GM}}{2a} +n_PI_P + \mathsf{R}(a,e,i,\omega,\Omega,u,f_P;a_P,e_P)
\] 
where $n_P$ is the planet's mean motion. The dependence of $\mathsf{H}$ on the modified Delaunay variables $(\Lambda,\Gamma,\Theta,\lambda,\gamma,\theta)$ is implicit, through the orbital elements, and the dependence on $\lambda_P$ is also implicit, through $f_P$.

\paragraph{3 - Expansion of the semi-major axis}
A key element of our proposed method is the following: for algorithmic convenience purposes, it turns out quite useful to have constant frequencies appearing at the kernel of the homological equation to be solved at successive normalization steps. This can be achieved in the following way: recalling that
\[
a = \frac{\Lambda^2}{\mathcal{GM}}, 
\]
the particle's semi-major axis $a$ can be expanded as
\begin{equation}
	a = \aV  + \frac{2}{n^*\aV}\delta \Lambda + \dots , \quad \mbox{with} \quad n^*= \sqrt{\frac{\mathcal{GM}}{{\aV}^{3}}}.
	\label{smaexp}
\end{equation}
Then, the Keplerian term in the Hamiltonian becomes: 
\begin{equation*}
-\frac{\mathcal{GM}}{2a} = -\frac{\mathcal{GM}}{2\aV} +n^*\delta \Lambda-\frac{3}{2}\frac{\delta \Lambda^2}{{\aV}^2}+\dots,
\end{equation*}
where, the term constant in $\delta \Lambda$ can be omitted. Introducing the above expansion, the Hamiltonian takes the form
\begin{equation*}
	\mathcal{H}= n_PI_P + n^*\delta \Lambda -\frac{3}{2}\frac{\delta \Lambda^2}{\aV{^2}} + \dots + \mathsf{R}(\aV + \frac{2}{n^*\aV}\delta\Lambda + \dots, e, i, \omega, \Omega, u, f_P; a_P, e_P),
\end{equation*}
with
 \begin{align*}
 & e =  \sqrt{1-\Big(1-\frac{\Gamma}{n^*{\aV}^2+\delta \Lambda}\Big)^2} & \omega = -\gamma + \theta \\
 & i = \arccos\Big(1-\frac{\Theta}{n^*{\aV}^2+\delta \Lambda -\Gamma}\Big) &\Omega=-\theta.
\label{MDdef}
\end{align*}
The angle $u$ depends on the canonical variables, $u = u(\delta \Lambda,\Gamma,\lambda,\gamma)$, through Kepler's equation
\[
M = u-e\sin(u)
\]
where $M=\lambda+\gamma$ is the mean anomaly. We note that this expansion of the Hamiltonian in powers of $\delta \Lambda$ is equivalent to the canonical transformation $\Big(\Lambda,\Gamma,\Theta,I_P,\lambda,\gamma,\theta,\lambda_P\Big)\rightarrow\Big(\delta\Lambda,\Gamma,\Theta,I_P,\lambda,\gamma,\theta,\lambda_P\Big)$.

\paragraph{4 - RM-reduction}
To the Hamiltonian found in the previous step, we perform the following operation, called   `${r}$-to the minus one' or RM-reduction, which does not alter $\mathcal{H}$: 
\begin{equation}
\mathcal{H}= n_PI_P + n^*\delta \Lambda + \Big(-\frac{3}{2}\frac{\delta \Lambda^2}{\aV{^2}} + \dots + \mathsf{R}(\aV + \frac{2}{n^*\aV}\delta \Lambda + \dots, e, i, \omega, \Omega, u, f_P; a_P, e_P)\Big)Q
\label{H3D}
\end{equation}
where 
\begin{equation}
Q = \frac{a(1-e\cos u)}{r}=\frac{\aV(1-e\cos u)}{r}+2\frac{(1-e\cos u)}{\aV n^* r}\delta \Lambda + ...=1.
\label{Qdef}
\end{equation}
The trigonometric reduction \eqref{H3D} yields a sum of trigonometric monomials $\cos(k_1u+k_2f_P+k_3\omega+k_4\Omega)$; moreover, after RM-reduction all terms in $\mathcal{H}$ appear divided by $r$ except for the terms $n_PI_P$ and $n^*\delta \Lambda$.

\subsection{Book Keeping}
A book-keeping symbol $\epsilon$, with numerical value $\epsilon=1$, is used in order to keep track of the relative size of the various terms in the Hamiltonian. 
There are four different small parameters to consider in the problem: $\mu$, $\delta\Lambda$ and the two eccentricities $e$ and $e_P$.
We adopt the following  `book-keeping rules' to assign a unique power of the symbol $\epsilon$ (reflecting the order of smallness) to each term in the Hamiltonian:
\begin{itemize}
	\item all terms depending on powers of the eccentricities $e^je_P^k$, with  $j,k\in\mathbb{Z}$, are multiplied by the book-keeping factor $\epsilon^{(j+k)}$;
	\item all terms depending on $(1+\eta)^j$ and $(1-\eta)^k$, with $j,k\in\mathbb{N}$, are multiplied by $\epsilon^0$ and $\epsilon^{2k}$ respectively; 
	\item all terms depending on $\mu^j\delta\Lambda^k$, with $j,k\in\mathbb{N}$, are multiplied by $\epsilon^{(j+k)s_0}$ with $s_0$ given in \eqref{s0def}; 
	\item all terms depending on $\delta \Lambda^k$, with $k\in\mathbb{N}$, coming from the Keplerian contribute in the Hamiltonian, are multiplied by $\epsilon^{(k-1)s_0}$; 
	\item all terms depending on $\phi^k$, with $k\in\mathbb{N}$, are multiplied by $\epsilon^k$. 
\end{itemize} 
The quantity $\phi=u-M$ is called  `equation of the center'. By Kepler's equation, we have $\phi=e\sin u$. 
After the assignment of the above book-keeping factors, the Hamiltonian is split into two main components,i.e. a leading term $Z_0$ and the disturbing function ${R}$, where 
\begin{equation}
Z_0=n^*\delta \Lambda+n_PI_P, \qquad {R}=\sum_{s=s_0}^{+\infty} \epsilon^s {R}_s^{(0)}.
\label{Z0def}
\end{equation}
To perform the above operation, and in particular to specify the value of the lowest book-keeping order $s_0$ in the perturbation, we must have an estimate of the size of $e$ along any individual trajectory: in the numerical examples below we use the initial value $e(t_0)$ for this purpose. 

\paragraph{}
In terms of the above book-keeping, the goal of the normalization becomes, now,  to define a Lie series transformation leading to a final Hamiltonian normalized up to a pre-selected order $s_m$ in the book-keeping parameter $\epsilon$. In particular, after $s_m-s_0+1$ normalization steps, the Hamiltonian will have the form:
\begin{equation*}
\mathcal{H}^{(s_m-s_0+1)}=Z(\delta \Lambda,\Gamma,\Theta,I_P,\gamma,\theta)+\sum_{s=s_m+1}^{+\infty}\epsilon^{s}R_s^{(s_m-s_0+1)}, \qquad s_m>s_0
\end{equation*}
where ${Z}$ is in normal form. All terms with book-keeping order higher than $s_m$ are considered negligible in the initial Hamiltonian. Then, the starting Hamiltonian for computing the normal form is set as:  
\begin{equation}
\mathcal{H} \sim \mathcal{H}^{(0)} = Z_0 + \sum_{s=s_0}^{s_m} \epsilon^s {R}_s^{(0)}.
\label{Rsum} 
\end{equation}
The order $s_m$ is called the maximum truncation order of the expansion. Let us remark that if we target a remainder with a size of order $\big(\frac{{m}_P}{\mathcal{M}}\big)^2$, we must impose 
\begin{equation}
	s_m=2s_0-1.
	\label{maxorder}
\end{equation}.

For a remainder of order $\big(\frac{{m}_P}{\mathcal{M}}\big)^{\ell}$, with $\ell>2$, we have, instead, $	s_m=\ell s_0-1$.

\subsection{Poisson structure}
\label{sectionPB}
All along the normalization in closed form, we need to compute Poisson brackets of the form $\{A_1,A_2\}$, where $A_1$ and $A_2$ are functions of $(\delta\Lambda,\Gamma,\Theta,I_P,\lambda,\gamma,\theta,\lambda_P)$ whose explicit expressions are given in terms of the orbital elements $(e,i,\omega,\Omega,u,f_P)$ and of the variables $r,\phi,\eta$: 

\begin{equation}
\begin{split}
A_{1,2}=&A(\delta \Lambda,e(\delta \Lambda,\Gamma),\eta(\delta\Lambda,\Gamma),i(\delta \Lambda,\Gamma,\Theta),\omega(\gamma,\theta),\Omega(\theta),\\ &u(\delta \Lambda,\Gamma,\lambda,\gamma),\phi(\delta\Lambda,\Gamma,\lambda,\gamma),r(\delta\Lambda,\Gamma,\lambda,\gamma),I_P,f_P(\lambda_P); \aV,a_P,e_P).
\end{split}
\label{formPBF}
\end{equation} 
To compute $\{A_1,A_2\}$ we use the formula
\begin{equation}
\begin{split}
\{A_1,A_2\}= & \frac{\partial A_1}{\partial \lambda}\frac{\partial A_2}{\partial \delta \Lambda}-\frac{\partial A}{\partial \delta \Lambda}\frac{\partial A_2}{\partial \lambda}+\frac{\partial A_1}{\partial \gamma}\frac{\partial A_2}{\partial \Gamma}-\frac{\partial A_1}{\partial \Gamma}\frac{\partial A_2}{\partial \gamma}+\frac{\partial A_1}{\partial \theta}\frac{\partial A_2}{\partial \Theta}-\frac{\partial A_1}{\partial \Theta}\frac{\partial A_2}{\partial \theta}\\
& +\Big(\frac{\partial A_1}{\partial \lambda_P}\frac{\partial A_2}{\partial I_P}-\frac{\partial A_1}{\partial I_P}\frac{\partial A_2}{\partial \lambda_P}\Big)\Big(\frac{\aV(1-e\cos u)}{r}+\mathcal{O}(\epsilon^{s_0}\delta \Lambda)\Big),
\end{split}
\label{pbdef}
\end{equation}
where the following formulas are applied for any $A=A_{1,2}$: 
\[
\frac{\partial A}{\partial \lambda} = \frac{\partial A}{\partial u}\frac{\partial u}{\partial \lambda}+\frac{\partial A}{\partial r}\frac{\partial r}{\partial \lambda}+\frac{\partial A}{\partial \phi}\frac{\partial \phi}{\partial \lambda}\epsilon^{-1},
\]
\[
\begin{split}
\frac{\partial A}{\partial \delta \Lambda} = & \frac{\partial A}{\partial \delta \Lambda}+\frac{\partial A}{\partial e} \frac{\partial e}{\partial \delta \Lambda}\epsilon^{-1}+\frac{\partial A}{\partial \eta}\frac{\partial \eta}{\partial \delta \Lambda}+\frac{\partial A}{\partial \cos i}\frac{\partial \cos i}{\partial \delta \Lambda}+\frac{\partial A}{\partial \sin i}\frac{\partial \sin i}{\partial \delta \Lambda}+\frac{\partial A}{\partial u}\frac{\partial u}{\partial \delta \Lambda}+\\ &\frac{\partial A}{\partial r}\frac{\partial r}{\partial \delta \Lambda}+\frac{\partial A}{\partial \phi}\frac{\partial \phi}{\partial \delta \Lambda}\epsilon^{-1},
\end{split}
\]
\[
\frac{\partial A}{\partial \gamma} = \frac{\partial A}{\partial \omega}\frac{\partial \omega}{\partial \gamma} + \frac{\partial A}{\partial M}\frac{\partial M}{\partial \gamma} +\frac{\partial A}{\partial u}\frac{\partial u}{\partial \gamma}+\frac{\partial A}{\partial r}\frac{\partial r}{\partial \gamma}+\frac{\partial A}{\partial \phi}\frac{\partial \phi}{\partial \gamma}\epsilon^{-1}
\]
\[
\begin{split}
\frac{\partial A}{\partial 	\Gamma} = & \frac{\partial A}{\partial u}\frac{\partial u}{\partial \Gamma}+\frac{\partial A}{\partial r}\frac{\partial r}{\partial \Gamma}+\frac{\partial A}{\partial \phi}\frac{\partial \phi}{\partial \Gamma}\epsilon^{-1}+\frac{\partial A}{\partial e}\frac{\partial e}{\partial \Gamma}\epsilon^{-1}+\frac{\partial A}{\partial \eta}\frac{\partial \eta}{\partial \Gamma}\\ &+ \frac{\partial A}{\partial \cos(i)}\frac{\partial \cos(i)}{\partial \Gamma}+\frac{\partial A}{\partial \sin(i)}\frac{\partial \sin(i)}{\partial \Gamma},
\end{split}
\]
\[
\frac{\partial A}{\partial \theta} = \frac{\partial A}{\partial \Omega}	\frac{\partial\Omega}{\partial\theta}+\frac{\partial A}{\partial \omega}	\frac{\partial\omega}{\partial\theta},
\]
\[
\frac{\partial A}{\partial \Theta} = \frac{\partial A}{\partial \cos(i)}\frac{\partial \cos(i)}{\partial \Theta}+\frac{\partial A}{\partial \sin(i)}\frac{\partial \sin(i)}{\partial \Theta},
\]
\[
\frac{\partial A}{\partial \lambda_P} = \frac{\partial A}{\partial f_P}\frac{\partial f_P}{\partial \lambda_P}. 
\]
The partial derivatives in the formulas above are:
\[
\begin{split}
\frac{\partial e}{\partial \delta \Lambda}= -\frac{\eta e}{(1+\eta)n^*{\aV}^2}\epsilon +\mathcal{O}(\epsilon^{s_0}\delta \Lambda),  &\quad \frac{\partial e}{\partial \Gamma}=-\frac{\eta}{a^{*^2} n^* e}	\epsilon^{-1} +\mathcal{O}(\epsilon^{s_0}\delta\Lambda),\\
\frac{\partial \eta}{\partial \delta \Lambda}= \frac{1-\eta}{n^*{\aV}^2}\epsilon^2+\mathcal{O}(\epsilon^{s_0}\delta\Lambda), &\quad \frac{\partial \eta}{\partial \Gamma}=\frac{1}{a^{*^2} n^*}+\mathcal{O}(\epsilon^{s_0}\delta\Lambda),\\
\frac{\partial \cos i}{\partial \delta \Lambda}=\frac{1-\cos i}{a^{*^2} n^* \eta} +\mathcal{O}(\epsilon^{s_0}\delta), &\quad
\frac{\partial \sin i}{\partial \delta \Lambda}=-\frac{\cos i (1-\cos i)}{a^{*^2} n^* \eta \sin(i)} +\mathcal{O}(\epsilon^{s_0}\delta\Lambda),\\
\frac{\partial \cos i}{\partial \Gamma}=\frac{\cos i-1}{a^{*^2} n^* \eta} +\mathcal{O}(\epsilon^{s_0}\delta\Lambda), 
&\quad \frac{\partial \sin i}{\partial \Gamma}=\frac{\cos i (1-\cos i)}{a^{*^2} n^* \eta \sin(i)} +\mathcal{O}(\epsilon^{s_0}\delta\Lambda),\\
\frac{\partial  \cos i}{\partial \Theta}=-\frac{1}{a^{*^2} n^* \eta} +\mathcal{O}(\epsilon^{s_0}\delta\Lambda), &\quad \frac{\partial  \sin i}{\partial 	\Theta}=\frac{\cos i}{a^{*^2} n^* \eta \sin i} +\mathcal{O}(\epsilon^{s_0}\delta\Lambda),\\
\frac{\partial u}{\partial 	\lambda}=\frac{a^*}{r}+\mathcal{O}(\epsilon^{s_0}\delta\Lambda), &\quad \frac{\partial u}{\partial 	\gamma}=\frac{a^*}{r}+\mathcal{O}(\epsilon^{s_0}\delta\Lambda),\\
\frac{\partial u}{\partial \delta \Lambda}=\frac{\eta e\sin u}{r(1+\eta)n^*{\aV}}\epsilon+\mathcal{O}(\epsilon^{s_0}\delta\Lambda), & \quad \frac{\partial u}{\partial \Gamma}=\frac{\eta\sin(u)}{a^*n^*er}\epsilon^{-1}+\mathcal{O}(\epsilon^{s_0}\delta\Lambda),\\
\frac{\partial \phi}{\partial \lambda}=\frac{a^*}{r}-1+\mathcal{O}(\epsilon^{s_0}\delta\Lambda), & \quad \frac{\partial \phi}{\partial \gamma}= \frac{a^*e\cos u}{r}\epsilon+\mathcal{O}(\epsilon^{s_0}\delta\Lambda),\\
 \frac{\partial \phi}{\partial \delta \Lambda}=\frac{\eta e\sin u}{r(1+\eta)n^*{\aV}}\epsilon+\mathcal{O}(\epsilon^{s_0}\delta\Lambda), & \quad \frac{\partial \phi}{\partial \Gamma}=\frac{\eta\sin(u)}{a^*n^*er}\epsilon^{-1}+\mathcal{O}(\epsilon^{s_0}\delta\Lambda), \\
 \frac{\partial r}{\partial \lambda}=\frac{{\aV}^2e\sin(u)}{r}\epsilon +\mathcal{O}(\epsilon^{s_0}\delta\Lambda), & \quad \frac{\partial r}{\partial \gamma}=\frac{{\aV}^2e\sin(u)}{r}\epsilon+\mathcal{O}(\epsilon^{s_0}\delta\Lambda),\\
  \frac{\partial r}{\partial \delta \Lambda}= \frac{\eta e\cos u}{(1+\eta)n^*{\aV}}\epsilon+\mathcal{O}(\epsilon^{s_0}\delta\Lambda), & \quad \frac{\partial r}{\partial \Gamma}=\frac{\eta(e\epsilon-cos(u))}{n^*er}\epsilon^{-1}+\mathcal{O}(\epsilon^{s_0}\delta\Lambda),\\
  \frac{\partial \omega}{\partial \theta} = 1, &\quad \frac{\partial \omega}{\partial \gamma} = -1,\\
  \frac{\partial \Omega}{\partial \theta} = -1, &\quad
  \frac{\partial f_P}{\partial \lambda_P}=1+\frac{2e_P\cos(f_P)}{\eta_P^3}\epsilon+\Big(\frac{1}{\eta_P^3}-1\\ & \hspace{10mm}+\frac{e_P^2\cos^2(f_P)}{\eta_P^3}\Big)\epsilon^2.\\
\end{split}
\]

Whenever needed, the higher order terms $\mathcal{O}(\epsilon^{s_0}\delta\Lambda)$ in the above formulas are easy to obtain by computer algebra. 

\paragraph{}
In order to allow for various simplifications during the normalization process, the previous expressions are to be implemented in the manipulator in the exact form given above. Note also the explicit appearance of the book-keeping parameter $\epsilon$ in all expressions of the partial derivatives, which depend on  $e$, $e_P$ and $\delta \Lambda$. This is an essential element of the method: supposing that $A_1=\tilde{A}_1\epsilon^j$ and $A_2=\tilde{A}_2\epsilon^k$, $j,k\ge s_0$, the result of $\{A_1,A_2\}$ will not be of order $j+k$ in $\epsilon$, but will contain several terms with different powers of $\epsilon$. In particular, we have the following

\begin{proposition}
	Given two functions $A_1=\tilde{A}_1\epsilon^j$ and $A_2=\tilde{A}_2\epsilon^k$ with $\tilde{A}_1$ and $\tilde{A}_2$ in the form \eqref{formPBF}, the Poisson bracket $\{A_1,A_2\}$ generates terms whose minimum order in $\epsilon$ is equal to 
	\begin{itemize}
		\item $j+k-2$ if $j,k> s_0$; 
		\item $j+k-1$ if either $j=s_0$, $k>s_0$ or $j>s_0$, $k=s_0$; 		
		\item $j+k$ if $j=k=s_0$.
	\end{itemize} 
\label{pboutcomesorder}
\end{proposition}
\begin{proof}
	We analyse the powers in $\epsilon$ of the various terms produced in the Poisson bracket $\{A_1,A_2\}$. The partial derivatives with respect to $\theta$ and $\Theta$ do not introduce any order variation, thus 
	\[
	\frac{\partial A_1}{\partial \theta}\frac{\partial A_2}{\partial \Theta}-\frac{\partial A_1}{\partial \Theta}\frac{\partial A_2}{\partial \theta}\sim\epsilon^{j+k}. 
	\] 
	Similarly, 
	\[
	\frac{\partial A_1}{\partial \lambda_P}\frac{\partial A_2}{\partial I_P}-\frac{\partial A_1}{\partial I_P}\frac{\partial A_2}{\partial \lambda_P}\sim\epsilon^{j+k}.
	\] 
	In fact, the partial derivatives with respect to $\lambda_P$ generate terms of orders $j+k$ or higher. The same holds true for
	\[
	\frac{\partial A_1}{\partial \lambda}\frac{\partial A_2}{\partial \delta \Lambda}-\frac{\partial A_1}{\partial \delta\Lambda}\frac{\partial A_2}{\partial \lambda}\sim \epsilon^{k+j}
	\] 
	
	The only part of $\{A_1,A_2\}$ which produces terms of order lower than $j+k$ is  
	\[
	\frac{\partial A_1}{\partial \gamma}\frac{\partial A_2}{\partial \Gamma}-\frac{\partial A_1}{\partial \Gamma}\frac{\partial A_2}{\partial \gamma}. 
	\] 
	
	We have 
	\[
	\frac{\partial A_{m}}{\partial \gamma} = 	\frac{\partial A_{m}}{\partial u}\frac{\partial u}{\partial \gamma}+\frac{\partial A_{m}}{\partial \omega}\frac{\partial \omega}{\partial \gamma} + \frac{\partial A_{m}}{\partial \phi}\frac{\partial \phi}{\partial \gamma}+ \frac{\partial A_{m}}{\partial r}\frac{\partial r}{\partial \gamma}, 
	\]
	where 
	\[
	\frac{\partial A_{m}}{\partial u}\frac{\partial u}{\partial \gamma}+\frac{\partial A_{m}}{\partial \omega}\frac{\partial \omega}{\partial \gamma}+\frac{\partial A_{m}}{\partial \phi}\frac{\partial \phi}{\partial \gamma}\sim \epsilon^{l}, \quad  \frac{\partial A_{m}}{\partial r}\frac{\partial r}{\partial \gamma}\sim \epsilon^{l+1}, \]
	and 
	\[
	\frac{\partial A_{m}}{\partial \Gamma} = \frac{\partial A_{m}}{\partial e}\frac{\partial e}{\partial \Gamma}+	\frac{\partial A_{m}}{\partial u}\frac{\partial u}{\partial \Gamma} + \frac{\partial A_{m}}{\partial r}\frac{\partial r}{\partial \Gamma} + \frac{\partial A_{m}}{\partial \cos i}\frac{\partial \cos i}{\partial \Gamma}+\frac{\partial A_{m}}{\partial \sin i}\frac{\partial \sin i}{\partial \Gamma},
	\]
	where 
	\[
	 \frac{\partial A_{m}}{\partial e}\frac{\partial e}{\partial \Gamma} + \frac{\partial A_{m}}{\partial \phi}\frac{\partial \phi}{\partial \Gamma}\sim \epsilon^{l-2}, \hspace{1mm}\frac{\partial A_{m}}{\partial u}\frac{\partial u}{\partial \Gamma} + \frac{\partial A_{m}}{\partial r}\frac{\partial r}{\partial \Gamma}\sim \epsilon^{l-1}, \hspace{1mm} \frac{\partial A_{m}}{\partial \cos i}\frac{\partial \cos i}{\partial \Gamma}+\frac{\partial A_{m}}{\partial \sin i}\frac{\partial \sin i}{\partial \Gamma}\sim \epsilon^l
	\]
	with 
	\[
	m=1,2, \quad l= \left\{ \begin{array}{rcl}
	j & \mbox{for}
	& m=1 \\ k & \mbox{for} & m=2 \\
	\end{array}\right. .
	\]
	Then, the quantity
	\[
	\begin{split}
	&\Big(\frac{\partial A_1}{\partial u}\frac{\partial u}{\partial \gamma}+\frac{\partial A_1}{\partial \omega}\frac{\partial \omega}{\partial \gamma}+\frac{\partial A_1}{\partial \phi}\frac{\partial \phi}{\partial \gamma}\Big)\Big(\frac{\partial A_2}{\partial e}\frac{\partial e }{\partial \Gamma}+\frac{\partial A_{2}}{\partial \phi}\frac{\partial \phi}{\partial \Gamma}\Big)\\ & -\Big(\frac{\partial A_2}{\partial u}\frac{\partial u}{\partial \gamma}+\frac{\partial A_2}{\partial \omega}\frac{\partial \omega}{\partial \gamma}+\frac{\partial A_2}{\partial \phi}\frac{\partial \phi}{\partial \gamma}\Big)\Big(\frac{\partial A_1}{\partial e}\frac{\partial e}{\partial \Gamma}+\frac{\partial A_{1}}{\partial \phi}\frac{\partial \phi}{\partial \Gamma}\Big)
	\end{split}
	\]
	generates terms of order $j+k-2$, while the quantity
	\[
	\begin{split}
	&\Big(\frac{\partial A_1}{\partial u}\frac{\partial u}{\partial \gamma}+\frac{\partial A_1}{\partial \omega}\frac{\partial \omega}{\partial \gamma}+\frac{\partial A_1}{\partial \phi}\frac{\partial \phi}{\partial \gamma}\Big)\Big(\frac{\partial A_2}{\partial u}\frac{\partial u}{\partial \Gamma}+\frac{\partial A_2}{\partial r}\frac{\partial r}{\partial \Gamma}\Big)+\frac{\partial A_1}{\partial r}\frac{\partial r}{\partial \gamma}\Big(	
	\frac{\partial A_2}{\partial e}\frac{\partial e }{\partial \Gamma}+	\frac{\partial A_2}{\partial \phi}\frac{\partial \phi}{\partial \Gamma}\Big) \\ &-\Big(\frac{\partial A_2}{\partial u}\frac{\partial u}{\partial \gamma}+\frac{\partial A_2}{\partial \omega}\frac{\partial \omega}{\partial \gamma}+\frac{\partial A_2}{\partial \phi}\frac{\partial \phi}{\partial \gamma}\Big)\Big(\frac{\partial A_1}{\partial u}\frac{\partial u}{\partial \Gamma}+\frac{\partial A_1}{\partial r}\frac{\partial r}{\partial \Gamma}\Big)-\frac{\partial A_2}{\partial r}\frac{\partial r}{\partial \gamma}\Big(	
	\frac{\partial A_1}{\partial e}\frac{\partial e }{\partial \Gamma}+\frac{\partial A_1}{\partial \phi}\frac{\partial \phi }{\partial \Gamma}\Big)
	\end{split}
	\] 
	generates terms of order $j+k-1$.
	\paragraph{}
	If both $j,k>s_0$, the minimum book-keeping order encountered in the above derivatives is $j+k-2$ and the proof of Proposition \ref{pboutcomesorder} is completed. We will show, now, that the minimum book-keeping order becomes $j+k-1$ if either $j=s_0$, $k>s_0$ or $k=s_0$, $j>s_0$. To this end, the following lemma must be used: 
	\begin{lemma}
		Consider the angular variables
		\begin{equation}
		L_E = u+\omega+\Omega, \quad
		\bar{\omega}=\omega +\Omega, \quad
		L_{T,P} = f_P +\Omega_P+\omega_P, \quad \bar{\omega}_P=\omega_P+\Omega_P
		\label{transfLemma1}
		\end{equation}
		where $L_E$ is the eccentric longitude, $L_{T,P}$ is the planet's true longitude and $\bar{\omega}$ and $\bar{\omega}_P$ are the longitudes of the pericenter of the particle. The Hamiltonian function contains terms of the form
		\[
		f_{\bm{m}}(\delta\Lambda,i,\eta,r;\aV,a_P,e_P)e_P^p e^q\cos(m_1L_E+m_2L_{T,P}+m_3\bar{\omega}+m_4\Omega+m_5\bar{\omega}_P+m_6\Omega_P)
		\]
		fulfilling the following D'Alembert rules:

	\begin{flalign}
 & \hspace{5mm} \bullet \hspace{1.5mm}	m_1+m_2+m_3+m_4+m_5+m_6=0 \label{d'al1} &\\
 &  \hspace{5mm} \bullet \hspace{1.5mm}	q-|m_3| \hspace{1.5mm} \mbox{is nonnegative and positive} \label{d'al2} & \\ 
 & \hspace{5mm} \bullet \hspace{1.5mm}	p-|m_5| \hspace{1.5mm} \mbox{is nonnegative and positive}. \label{d'al3} 
\end{flalign}

\label{lemmaDR}
\end{lemma}
We recall that in the selected reference frame $\Omega_P=\omega_P=0$. The proof of Lemma \ref{lemmaDR} is given in Appendix \ref{Applemmaproof}.

If, now, $j=s_0$, $A_1$ does not depend explicitly on $e$ and $\phi$ (as a consequence of the adopted book-keeping rules). The only small parameter on which it can depend will be either $\mu$ or $\delta \Lambda^2$. Thus
\[
\frac{\partial A_1}{\partial e}=\frac{\partial A_1}{\partial \phi}=0. 
\]
Moreover, from Lemma \ref{lemmaDR} it follows that $A_1$  does not contain $\bar{\omega}$. This implies that either it does not depend on $u$ and $\omega$ or it is of the form 
\begin{equation*}	A_1 = \epsilon^{j}f_{\bm{k}}(\eta,i,r;e_P,a_P,\aV)\cos(k_1u+k_2f_P+k_3\omega+k_4\Omega), \qquad k_3=k_1.
\end{equation*}
We have 
\[
\begin{split}
\frac{\partial A_1}{\partial u}\frac{\partial u}{\partial \gamma}+\frac{\partial A_1}{\partial \omega}\frac{\partial \omega}{\partial \gamma}=  -\epsilon^{j} & f_{\bm{k}}(\eta,i,r;e_P,a_P,\aV)\sin(k_1u+k_2f_P+k_3\omega+k_4\Omega)\Big(k_1\frac{\aV}{r}-k_3\Big)\\ &+\epsilon^{s_0+j}\mathcal{O}(\delta \Lambda).
\end{split}
\]
Since $k_1=k_3$ and $r=\aV(1-\epsilon e\cos u)+\mathcal{O}(\epsilon^{s_0}\delta \Lambda)$,  it follows that
\[
k_1\frac{\aV}{r}-k_3=\frac{k_1}{r}\aV e \cos u \epsilon +\mathcal{O}(\epsilon^{s_0} \delta \Lambda).
\]
so that
	\[
	\frac{\partial A_1}{\partial u}\frac{\partial u}{\partial \gamma}+\frac{\partial A_1}{\partial \omega}\frac{\partial \omega}{\partial \gamma}\sim \epsilon^{j+1}.
	\]	
Hence, if $j=s_0$ and $k>s_0$ the terms of smallest order, equal to $j+k-1$, are generated by 
\[
\begin{split}
&\Big(\frac{\partial A_1}{\partial u}\frac{\partial u}{\partial \gamma}+\frac{\partial A_1}{\partial \omega}\frac{\partial \omega}{\partial \gamma}+\frac{\partial A_1}{\partial r}\frac{\partial r}{\partial \gamma}\Big)\Big(	
\frac{\partial A_2}{\partial e}\frac{\partial e }{\partial \Gamma}+\frac{\partial A_2}{\partial \phi}\frac{\partial \phi}{\partial \Gamma}\Big)\\ & -\Big(\frac{\partial A_2}{\partial u}\frac{\partial u}{\partial \gamma}+\frac{\partial A_2}{\partial \omega}\frac{\partial \omega}{\partial \gamma}+\frac{\partial A_2}{\partial \phi}\frac{\partial \phi}{\partial \gamma}\Big)\Big(\frac{\partial A_1}{\partial u}\frac{\partial u}{\partial \Gamma}+\frac{\partial A_1}{\partial r}\frac{\partial r}{\partial \Gamma}\Big)
\end{split}
\] 

In the same way we prove that when $k=s_0$ and $j>s_0$ the minimum order of $\{A_1,A_2\}$ is limited from below by $j+k-1$.  If $j=s_0$, $k>s_0$ or $k=s_0$, $j>s_0$, this completes the proof. If, however, both $k=j=s_0$, by the same formulas we have that $\{A_1,A_2\}$ does not generate any terms of order lower than $j+k$.  
\end{proof}
\paragraph{}
\begin{remark}
	Given two functions $A_1=\tilde{A}_1\epsilon^{j}$ and $A_2=\tilde{A}_2\epsilon^k$, with $j,k>s_0$ and $\tilde{A}_1$, $\tilde{A}_2$ of the form \eqref{formPBF}, the part of 
	$\{A_1,A_2\}$ generating terms of order $j+k-2$ is
	\begin{equation}
	\begin{split}
			&	\Big(\frac{\partial A_1}{\partial u}\frac{\partial u}{\partial \gamma}+\frac{\partial A_1}{\partial \omega}\frac{\partial \omega}{\partial \gamma}+\frac{\partial A_1}{\partial \phi}\frac{\partial \phi}{\partial \gamma}\Big)\Big(\frac{\partial A_2}{\partial e}\frac{\partial e }{\partial \Gamma}+\frac{\partial A_{2}}{\partial \phi}\frac{\partial \phi}{\partial \Gamma}\Big)\\ &-\Big(\frac{\partial A_2}{\partial u}\frac{\partial u}{\partial \gamma}+\frac{\partial A_2}{\partial \omega}\frac{\partial \omega}{\partial \gamma}+\frac{\partial A_2}{\partial \phi}\frac{\partial \phi}{\partial \gamma}\Big)\Big(\frac{\partial A_1}{\partial e}\frac{\partial e}{\partial \Gamma}+\frac{\partial A_{1}}{\partial \phi}\frac{\partial \phi}{\partial \Gamma}\Big)
	\end{split}
	\label{c_orderm2}
	\end{equation}
	\label{remarkOrderM2}
\end{remark}

\begin{remark}
	Given two functions $A_1=\tilde{A}_1\epsilon^{s_0}$ and $A_2=\tilde{A}_2\epsilon^k$, $k>s_0$, with $\tilde{A}_1$ and $\tilde{A}_2$ of the form \eqref{formPBF}, the part of  $\{A_1,A_2\}$ generating terms of order $s_0+k-1$ is 
	\begin{equation}
	\begin{split}
&	\Big(\frac{\partial A_1}{\partial u}\frac{\partial u}{\partial \gamma}+\frac{\partial A_1}{\partial \omega}\frac{\partial \omega}{\partial \gamma}+\frac{\partial A_1}{\partial r}\frac{\partial r}{\partial \gamma}\Big)\Big(	
	\frac{\partial A_2}{\partial e}\frac{\partial e }{\partial \Gamma}+	\frac{\partial A_2}{\partial \phi}\frac{\partial \phi}{\partial \Gamma}\Big) \\ & -\Big(\frac{\partial A_2}{\partial u}\frac{\partial u}{\partial \gamma}+\frac{\partial A_2}{\partial \omega}\frac{\partial \omega}{\partial \gamma}+\frac{\partial A_2}{\partial \phi}\frac{\partial \phi}{\partial \gamma}\Big)\Big(\frac{\partial A_1}{\partial u}\frac{\partial u}{\partial \Gamma}+\frac{\partial A_1}{\partial r}\frac{\partial r}{\partial \Gamma}\Big)
	\end{split}
	\label{c_orderm1}
	\end{equation}
	\label{remarkCriticalTerms}
\end{remark}

 \begin{remark}
	Given two functions $A_1=\tilde{A}_1\epsilon^{j}$ and $A_2=\tilde{A}_2\epsilon^k$, with $j,k>s_0$ and $\tilde{A}_1$, $\tilde{A}_2$ of the form \eqref{formPBF}, if $\tilde{A}_1$ does not explicitly depend on the equation of the center $\phi$ and on the eccentricity $e$ and it is either of the form
	\begin{equation}
	\tilde{A}_1 = f_{\bm{k}}(\eta,i,r;e_P,a_P,\aV)\cos(k_1u+k_2f_P+k_3\omega+k_4\Omega), \qquad k_3=k_1, k_1\ge 0
	\label{A1k1k3cos}
	\end{equation}
	or 
	\begin{equation}
	\tilde{A}_1 = f_{\bm{k}}(\eta,i,r;e_P,a_P,\aV)\sin(k_1u+k_2f_P+k_3\omega+k_4\Omega), \qquad k_3=k_1, k_1\ge 0
	\label{A1k1k3sin}
	\end{equation}
	the Poisson bracket $\{A_1,A_2\}$ produces terms of order equal or larger than $j+k-1$ (since the quantity in \eqref{c_orderm2} is zero). 
	\label{k1k3sameremark}
\end{remark}

\begin{remark}
	To automatically obtain all the terms with the correct book-keeping order, all terms generated by expressions of the form $\frac{\partial A_1}{\partial \gamma}\frac{\partial A_2}{\partial \Gamma}-\frac{\partial A_1}{\partial \Gamma}\frac{\partial A_2}{\partial \gamma}$ are automatically adjusted to appear with the same exponent of $r$ in the denominator. Then, in the numerator of all the resulting terms the variable $r$ is substituted with its expansion $r=\aV(1-\epsilon e \cos u)+\mathcal{O}(\epsilon^{s_0}\delta\Lambda)$.
\end{remark}

\begin{remark}
	If we target a normal form of order $s_m$, as defined in \eqref{maxorder}, all the contributions  $\mathcal{O}(\epsilon^{s_0}\delta \Lambda)$ in the partial derivatives used to compute the Poisson bracket can be neglected.
\end{remark}

\begin{remark}
	Once having computed the Poisson bracket by applying formula \eqref{pbdef}, we substitute $\phi$ with $ e\sin u$  in all produced terms depending on the equation of the center.
\end{remark}

\subsection{Homological Equation}
\label{section: normalization}
As mentioned in the introduction, at the $j$-th iteration of the normalization process, we must determine a generating function $\chi_{s_0+j-1}^{(j)}$ satisfying a homological equation of the form
\begin{equation}
\{Z_0,\chi_{s_0+j-1}^{(j)}\}+\epsilon^{s_0+j-1}{R}_{s_0+j-1}^{(j-1)}= \epsilon^{s_0+j-1}Z_{s_0+j-1}+ \mathcal{O}(\epsilon^{s_0+j}).
\label{homeq}
\end{equation}
We now give the precise form of the homological equation. 
\paragraph{}
By applying the formulas of subsection \ref{sectionPB}, the Poisson bracket $\{Z_0,\chi_{s_0+j-1}^{(j)}\}$ is given by 
\begin{equation}
\begin{split}
\{Z_0,\chi_{s_0+j-1}^{(j)}\} = & -n^*\Bigg(\frac{\aV}{r}\frac{\partial \chi_{s_0+j-1}^{(j)}}{\partial u}+\Big(\frac{\aV}{r}-1\Big)\frac{\partial \chi_{s_0+j-1}^{(j)}}{\partial \phi}\epsilon^{-1}+\frac{{\aV}^2 e\sin u }{r}\epsilon \frac{\partial \chi_{s_0+j-1}^{(j)}}{\partial r}\Bigg)\\ 
& -n_P\Bigg(1+ \frac{2e_P\cos(f_P)}{\eta_P^3}\epsilon +  \Big(\frac{1}{\eta_P^3}+\frac{e_P^2\cos^2(f_P)}{\eta_P^3}-1\Big)\epsilon^2\Bigg)\Bigg(\frac{\aV(1-\epsilon e\cos u)}{r}
\\&+\mathcal{O}(\epsilon^{s_0}\delta \Lambda)\Bigg)\frac{\partial \chi_{s_0+j-1}^{(j)}}{\partial f_P}.
\end{split}
\label{pbexpr}
\end{equation}

We then define $\chi_{s_0+j-1}^{(j)}$ by solving the equation: 
\begin{equation}
\begin{split}
-n^*\Bigg(\frac{\aV}{r}\frac{\partial \chi_{s_0+j-1}^{(j)}}{\partial u}+\Big(\frac{\aV}{r}-1\Big)\frac{\partial \chi_{s_0+j-1}^{(j)}}{\partial \phi}\epsilon^{-1}\Bigg) -n_P\frac{\aV}{r} &\frac{\partial \chi_{s_0+j-1}^{(j)}}{\partial f_P} +\epsilon^{s_0+j-1}{R}_{s_0+j-1}^{(j-1)}\\ &= \epsilon^{s_0+j-1}Z_{s_0+j-1}.
\end{split}
\label{chieq}
\end{equation}
The solution of \eqref{chieq} is found as follows: the function $Z_{s_0+j-1}$ contains all the terms of ${R}_{s_0+j-1}^{(j-1)}$ not depending on $\lambda$ and $\lambda_P$. Beside these terms, the function ${R}_{s_0+j-1}^{(j-1)}$ contains four more different types of terms:

\begin{itemize}
	\item type 1:  $\frac{\aV}{r}f(e,i,\eta,\omega,\Omega)$,
	\item type 2:  $ \frac{\aV}{r} \widehat{f}_{\bm{k}}(e,i,\eta)\cos(k_1u+k_2f_P+k_3\omega+k_4\Omega)$,
	\item type 3:  $\frac{\aV}{r^p}\bar{f}(e,i,\eta,\omega,\Omega)$, $p>1$,
	\item type 4:  $ \frac{\aV}{r^p} \tilde{f}_{\bm{k}}(e,i,\eta)\cos(k_1u+k_2f_P+k_3\omega+k_4\Omega)$, $p>1$.
\end{itemize} 
 
Depending on the type of encountered term to be normalized, the generating function $\chi_{s_0+j-1}^{(j)}$ must acquire a corresponding term equal to:
\begin{itemize}
	\item for type 1: $\epsilon\frac{1}{n^*}f(e,i,\eta,\omega,\Omega)\phi$,
	\item for type 2:  $ \frac{1}{k_1n^*+k_2n_P} \widehat{f}_{\bm{k}}(e,i,\eta)\cos(k_1u+k_2f_P+k_3\omega+k_4\Omega)$,
	\item for type 3:  $\epsilon\frac{1}{n^*}\phi\sum_{k=1}^{p}\frac{\hat{f}(e,\eta,i,\omega,\Omega)}{a^{k-1}r^{p-k}}$, $p>1$,
	\item for type 4:  $\frac{1}{k_1n^*+k_2n_P}\frac{1}{{r}^{p-1}} {\tilde{f}_{\bm{k}}(e,\eta,i)\sin(k_1u+k_2f_P+k_3\omega+k_4\Omega)}$, $p>1$.
\end{itemize} 
Then, the outcome of the operation $\{Z_0,\chi_{s_0+j-1}^{(j)}\}+\epsilon^{s_0+j-1}R_{s_0+j-1}^{(j-1)}$ yields terms in the normal form having the form as follows:
\begin{itemize}
	\item for each normalized term of type 1:  $f(e,i,\eta,\omega,\Omega)$,
	\item for each normalized term of type 2:  $0$,
	\item for each normalized term of type 3:  $\frac{\bar{f}(e,i,\eta,\omega,\Omega)}{{\aV}^{p-1}}$, $p>1$,
	\item for each normalized term of type 4:  $ 0$.
\end{itemize}  
We note that the residual of the normalization is equal to zero only for the terms of type~1. 
Another important remark regards the average value $<\chi_{s_0+j-1}^{(j)}>$ of the generating function $\chi_{s_0+j-1}^{(j)}$ with respect to the angles $\lambda$,$\lambda_P$. We have that
\[
\begin{split}
\chi_{s_0+j-1}^{(j)}= & \epsilon^{s_0+j-1}\sum_{\bm{k}} \frac{1}{k_1n^*+k_2n_P}\frac{f_{\bm{k}}(e,\eta,i)}{r^q}\sin(k_1u+k_2f_P+k_3\omega+k_4\Omega) \\ & +\epsilon^{s_0+j}\sum_m \frac{1}{n^*}\frac{f_m(e,i,\eta,\omega,\Omega)}{r^p}\phi, \quad p,q\ge1.
\end{split}
\]
Thus, the average $<\chi_{s_0+j-1}^{(j)}>$ is different from zero. This generates no problem for the iterative application of the method. However, it is customary to subtract from $\chi_{s_0+j-1}^{(j)}$ the average $<\chi_{s_0+j-1}^{(j)}>$ in order that the elements found in the normal form properly correspond to mean elements (see \cite{Lara2013_2}). We collect in Appendix \ref{AppAV} all the formulas required for the computation of the average $<\chi_{s_0+j-1}^{(j)}>$. 

\paragraph{}
From Proposition \ref{pboutcomesorder}, we have that the terms generated by the Lie transformation are of order higher than the term normalized at each step when $s_0>1$. In case $s_0=1$ this no longer holds true. We, then, have two distinct algorithms to perform the normalization depending on whether $s_0>1$ or $s_0=1$. 

\subsection{Normalization process for $s_0>1$} 
The normalization process consists of determining a succession of Lie transformations leading to the targeted normal form. If $s_m$ is the targeted order of the final normal form, $s_m-s_0+1$ steps must be performed. At each step $j$ the goal is to normalize the Hamiltonian $H^{(j-1)}$ obtained at the previous step. For this purpose, the homological equation (equations \eqref{homeq} and \eqref{chieq}) is solved to determine the generating function $\chi_{s_0+j-1}^{(j)}$; the new Hamiltonian is 
\[
\mathcal{H}^{(j)}=\exp(\mathcal{L}_{\chi_{s_0+j-1}^{(j)}})\mathcal{H}^{(j-1)}=Z_0+\sum_{s=s_0}^{s_0+j-1}\epsilon^sZ_{s}+\sum_{s=s_0+j}^{s_m}\epsilon^s{{R}}_{s}^{(j)};
\]
For $j=1$ (first step), $\mathcal{H}^{(j-1)}=\mathcal{H}^{(0)}$ (see \eqref{Rsum}).

\paragraph{} 
The remainder terms ${{R}}_{(s)}^{(j)}$, with $s\ge s_0+j$, contain three parts:
\begin{enumerate}[i)]
	\item $R_{s}^{(j-1)}$;
	\item the remainder of the homological equation $\eqref{homeq}$;
	\item the terms generated by the the Lie transformation, i.e. coming from  \[
	\Big\{\sum_{s=s_0}^{s_0+j-2}\epsilon^sZ_{s}+\sum_{s=s_0+j-1}^{s_m}\epsilon^s{{R}}_{(s)}^{(j-1)},\chi_{s_0+j-1}^{(j)}\Big\} + \frac{1}{2}\Big\{\Big\{\mathcal{H}^{(j-1)},\chi_{s_0+j-1}^{(j)}\Big\},\chi_{s_0+j-1}^{(j)}\Big\}+\dots;
	\]
\end{enumerate}
Concerning the last part, from Proposition \ref{pboutcomesorder} we have
\[
\{R_{s}^{(0)},\chi_{s_0}^{(1)}\}=\left\{\begin{array}{cl} \mathcal{O}(\epsilon^{2s_0}), & s=s_0,\\
\mathcal{O}(\epsilon^{s+s_0-1}), & s>s_0,
\end{array}\right., \qquad \mbox{for} \quad j=1 \]
while 
\[
\begin{split}
\{R_{s}^{(j-1)},\chi_{s_0+j-1}^{(j)}\}= & \mathcal{O}(\epsilon^{s+s_0+j-3}), \quad s\ge s_0+j-1 \\
\{Z_{s},\chi_{s_0+j-1}^{(j)}\}= & \left\{\begin{array}{cl} \mathcal{O}(\epsilon^{s+s_0+j-2}), & s=s_0,\\
\mathcal{O}(\epsilon^{s+s_0+j-3}), & s_0<s<s_0+j-1
\end{array}\right.\\ \end{split},  \qquad \forall j>1. \]
Then, the smallest order of the terms coming from the Lie transformation is equal to $2s_0$, for $j=1$, or equal to $2s_0+j-2$, for $j>1$. Since $s_0>1$, we have that the remainder is always of order higher than $s_0+j-1$, i.e. the order of the normalized term in the Hamiltonian $\mathcal{H}^{(j-1)}$. 
\paragraph{}

A detailed example of the normalization process for $s_0>1$ is given in Appendix \ref{examplenormprocss}. We note that the case $s_0>1$ is rather generic, in the sense that it applies to all trajectories except for those with $e=\mathcal{O}({m}_P/\mathcal{M})$.
 
 \subsection{Normalization process for $s_0=1$}
 \paragraph{Size of the remainder}
 If the particle's orbital eccentricity $e$ is very small($e\sim\mathcal{O}({\frac{{m}_P}{\mathcal{M}}})$) we obtain from \eqref{s0def} $s_0=1$. In this case, at the generic $j$-th iteration of the normalization algorithm it is easy to see that the operator $\mathcal{H}^{(j)}=\exp(L_{\chi_j})\mathcal{H}^{(j-1)}$ produces remainder terms of the \textit{same} book-keeping order as those normalized. Consider the Poisson bracket
  \[
 \Big\{\sum_{s=1}^{j-1}\epsilon^sZ_{s}+\sum_{s=j}^{s_m}\epsilon^s{{R}}_{(s)}^{(j-1)},\chi_{j}^{(j)}\Big\} + \frac{1}{2}\Big\{\Big\{\mathcal{H}^{(j-1)},\chi_{j}^{(j)}\Big\},\chi_{j}^{(j)}\Big\}+\dots
 \]
 
 From Proposition \ref{pboutcomesorder} we have
 \[
 \{R_{s}^{(0))},\chi_{1}^{(1)}\}=\left\{\begin{array}{cl} \mathcal{O}(\epsilon^{2}), & s=1,\\
 \mathcal{O}(\epsilon^{s}), & s>1,
 \end{array}\right., \qquad \mbox{for} \quad j=1. \]
 However
 \[
 \begin{split}
 \{R_{s}^{(j-1)},\chi_{j}^{(j)}\}= &\mathcal{O}(\epsilon^{s+j-2}), \quad s\ge j\\
 \{Z_{s},\chi_{j}^{(j)}\}=&\left\{\begin{array}{cl} \mathcal{O}(\epsilon^{s+j-1}), & s=1,\\
 \mathcal{O}(\epsilon^{s+j-2}), & 1<s<j.
 \end{array}\right.\\
 \end{split}, \qquad \forall j>1
 \]
 Then, for $j=2$ the Poisson brackets $\{R_{2}^{(1)},\chi_{2}^{(2)}\}$ and $\{Z_{1},\chi_{2}^{(2)}\}$ generate terms of book-keeping order equal to $2$. Similarly, for $j>2$ the Poisson brackets $\{Z_{1},\chi_{j}^{(j)}\}$ and $\{Z_{2},\chi_{j}^{(j)}\}$ generate terms of order $j$.  Control of these extra terms can be achieved on the basis of the following   
 \begin{proposition}
 	The following properties hold
 	\begin{itemize}
 		\item[i)] The normal form term $Z_{1}$ satisfies the relation
 		\[
 		\{Z_{1},\chi_{j}^{(j)}\}\sim \epsilon^{j+1}, \quad \forall j\ge 2; 
 		\] 
 		\item[ii)] 
 		The normal form term ${Z_{2}}$ satisfies the relation 
 		\[
 		\{Z_{2},\chi_{j}^{(j)}\}\sim \epsilon^{j+1}, \quad  \forall j\ge 3; 
 		\] 
 		\item[iii)] Let $\{R_{2}^{(2)}\}$ denotes the terms of book-keeping order $2$ coming from $\{R_{2}^{(1)},\chi_{2}^{(2)}\}$. Let $\chi_{2}^{(2,\rm bis)}$  be the Lie generating function normalizing $R_{2}^{(2)}$ according to equation \eqref{homeq}. Let 
 		\[
 		\mathcal{H}^{(2,\rm bis)}=\exp(\mathcal{L}_{\chi_{2}^{(2,\rm bis)}})\mathcal{H}^{(2)}=Z_0+Z_1+Z_2+Z_{2,\rm bis}+\sum_{s=3}^{s_m}\epsilon^sR_s^{(2, \rm bis)}
 		\]
 		be the new Hamiltonian computed by the Lie transform with the generating function $\chi_{2}^{(2,\rm bis)}$. The remainder of $\mathcal{H}^{(2,\rm bis)}$ has terms of book-keeping order larger than $2$. 		Moreover, we have
 		\[
 		\{Z_{2, \rm bis},\chi_{j}^{(j)}\}\sim \epsilon^{j+1}, \quad \forall j\ge3.
 		\] 	
 	\end{itemize}
 	\label{PropositionImportant}
 \end{proposition}
 The proof of Proposition \ref{PropositionImportant} is given in appendix \ref{AppPropProof}. From it, it follows that only at the second step of the normalization process the Lie transformation will generate terms with the same order as the normalized term. We show now how to deal with this problem by performing just one more additional normalization step. 
 
\paragraph{Adjustment of the normalization process}
The normalization process must be modified as follows:

\begin{itemize}
	\item The first step is as in the case $s_0>1$.
	\item  The second step consists of two sub-steps; in the first sub-step, the generating function $\chi_{2}^{(2)}$ is determined leading to the new Hamiltonian
	\[
	\mathcal{H}^{(2)}=\exp(\mathcal{L}_{\chi_{2}^{(2)}})\mathcal{H}^{(1)}.
	\]
	In the second sub-step,  the generating function $\chi_{2}^{(2,\rm bis)}$ is computed as described above and the new Hamiltonian is
	\[
	\mathcal{H}^{(2,\rm bis)}=\exp(\mathcal{L}_{\chi_{2}^{(2,\rm bis)}})\mathcal{H}^{(2)}.
	\]
	\item In the third step, the Hamiltonian $\mathcal{H}^{(2,\rm bis)}$ is normalized up to the third order in $\epsilon$; the Lie transformation leads to the new Hamiltonian 
\[
\mathcal{H}^{(3)}=\exp(\mathcal{L}_{\chi_{3}^{(3)}})\mathcal{H}^{(2,\rm bis)}.
\]
\item Successive iterations beyond the order $3$ are performed as in the case $s_0>1$.
\end{itemize}

\section{Numerical application in the PCR3BP }
\label{section: outcomesPCR3BP}
We applied the method described in section \ref{section: method} in the case of the PCR3BP considering Jupiter as the perturbing planet. The orbital planes of the body and planet coincide and the planet orbit is assumed circular ($e_P=0$). This implies that $f_P=\lambda_P$, $i=0$ and $\Omega=0$ so that the Hamiltonian does not depend on the Delaunay variables $\Theta$ and $\theta$.
\paragraph{}
 We perform two tests to assess the applicability and precision of the method. 
As a first test, we estimate the size of reminder $\mathcal{H}_R=\sum_{s\ge s_m+1}^{+\infty} R_s^{(s_m-s_0+1)}$ of the normal form and compare it to the size of the initial disturbing function $R=\mathcal{H}^{(0)}-Z_0$ (see \eqref{Rsum}). We perform
a multipolar expansion of degree $10$ and a normalization up to a certain order $s_m$ in book-keeping set as
\[
s_m = \min(2s_0-1,s_0+10)
\]
with $s_0$ given in \eqref{s0def}. Such a choice is empirically found to yield a good compromise between computational load and requirements for precision.
\paragraph{}
 To obtain estimates of the remainder size, we consider a truncation of the remainder up to terms of book-keeping order $s_m+3$: 
\[
\tilde{\mathcal{H}}_R=\sum_{s= s_m+1}^{s_m+3} R_s^{(s_m-s_0+1)};
\]
Writing $\mathcal{H}_R$ in the form
\[
\tilde{\mathcal{H}}_R=\sum_{\bm{k}} \frac{\aV}{r^q}f_{\bm{k}}(e,\eta)\cos(k_{1}u+k_{2}\lambda_P+k_{3}\omega), \qquad q\ge 1,
\]
the size of $\mathcal{H}_R$ can be estimated through the norm
\[
||\tilde{\mathcal{H}}_R||=\sum_{\bm{k}} \Big|\frac{f_{\bm{k}}(e,\eta)}{{\aV}^{q-1}(1-e)^q}\Big|.
\]
The norm of $R$ was computed with the same definition. Figure \ref{fig: errorGraph} shows  $\log_{10}\Big(||\tilde{\mathcal{H}}_R||/||{{R}}||\Big)$ in color scale in a grid of values for the initial semi-major axis $a(0)=\aV$ and eccentricity $e(0)=\sqrt{1-\Big(1-\frac{\Gamma(0)}{n*{\aV}^2}\Big)^2}$. The quantity $\log_{10}\Big(||\tilde{\mathcal{H}}_R||/||{{R}}||\Big)$ gives an estimate of the relative size of the remainder with respect to the initial perturbation, which estimates, in turn, the relative error of the semi-analytically computed trajectory with respect to the true trajectory. Denoting $e_0=e(0)$ and $a_0=a(0)$, in Figure \ref{fig: errorGraph} the red line corresponds to the set of points $(a_0,e_0)$ such that $a_0(1+e_0)=a_P$, i.e. the radius at the apocenter coincides with the Sun-Jupiter distance (for the PR3BP $r_P=a_P$). Since our method is applicable to particles with trajectories lying entirely inside the trajectory of the planet, the red line represents an upper boundary of the region in the $(a_0,e_0)$ plane in which the method can be applied. The black line represents the upper boundary of the values $(a_0,e_0)$ for which we have Hill-stable orbits. An orbit is defined as Hill-stable when its Jacobi constant $C$ is larger than the Jacobi constant $C_{L_1}$ at the Lagrangian point $L_1$. The boundary drawn was determined as described in \cite{Ramos2015}. 

From Figure \ref{fig: errorGraph}, we can observe that the relative error is lower than $10^{-2}$ for each value of the eccentricity up to an initial semi-major axis lower than $\sim 0.3 a_P$; up to $\sim 0.45 a_P$ it raises above $10^{-1}$ only for the higher values of $e(0)$. On the other hand, getting closer to the Hill-unstable region, the error becomes higher, and keeps having acceptable values (of few percent) only in regions with low eccentricity. In the Hill-unstable region the error is everywhere high. In the figure we can notice also several vertical strips along which the error is always higher than in their neighbourhood. These strips correspond to mean motion resonances, in which the method fails due to small divisors appearing along the normalization process; note that in the Hill-stable region, the error is high at those domains where the concentration of these strips becomes more conspicuous.  

\begin{figure}[h!]
	\centering
	\includegraphics[width=0.6\textwidth]{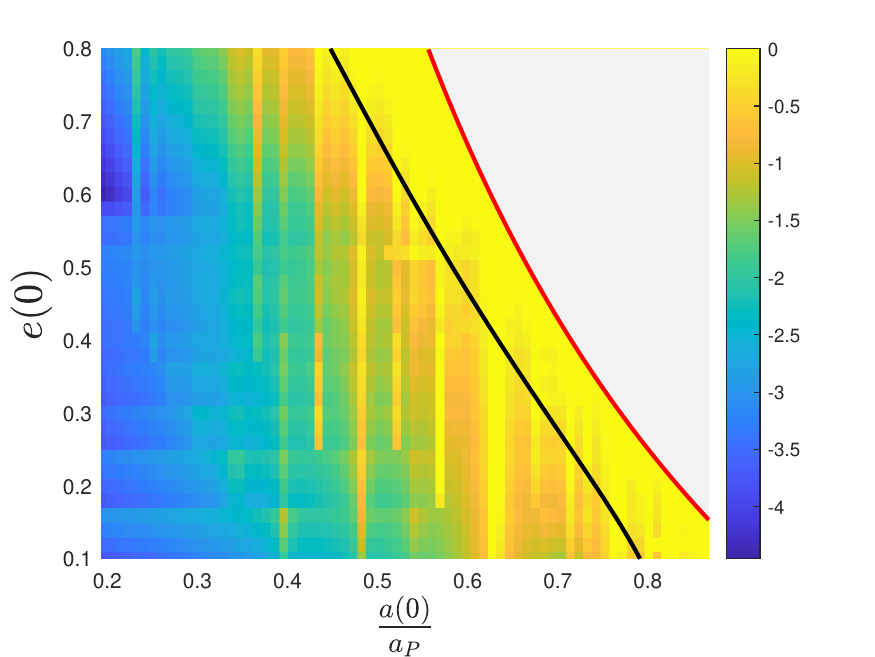}
	\caption{Variation of the relative error  $\log_{10}\Big(||\tilde{\mathcal{H}}_R||/||{{R}}||\Big)$ with respect to the initial values of the semi-major axis and the eccentricity: $||\tilde{\mathcal{H}}_R||$ is an upper bound estimate of the norm of the remainder after the normalization process;  $||R||$ is an upper bound estimate of the norm of the initial disturbing function, equal to the initial Hamiltonian minus the leading term. The red line is the set of point $(a,e)$ such that $a(1+e)=a_P$. On the left of the black line, the points $(a,e)$ correspond to Hill-stable trajectories. Further details are given in the text. }
	\label{fig: errorGraph}
\end{figure}

\begin{figure}
	\begin{subfigure}{0.48\textwidth}
		\centering
		\includegraphics[width=\textwidth]{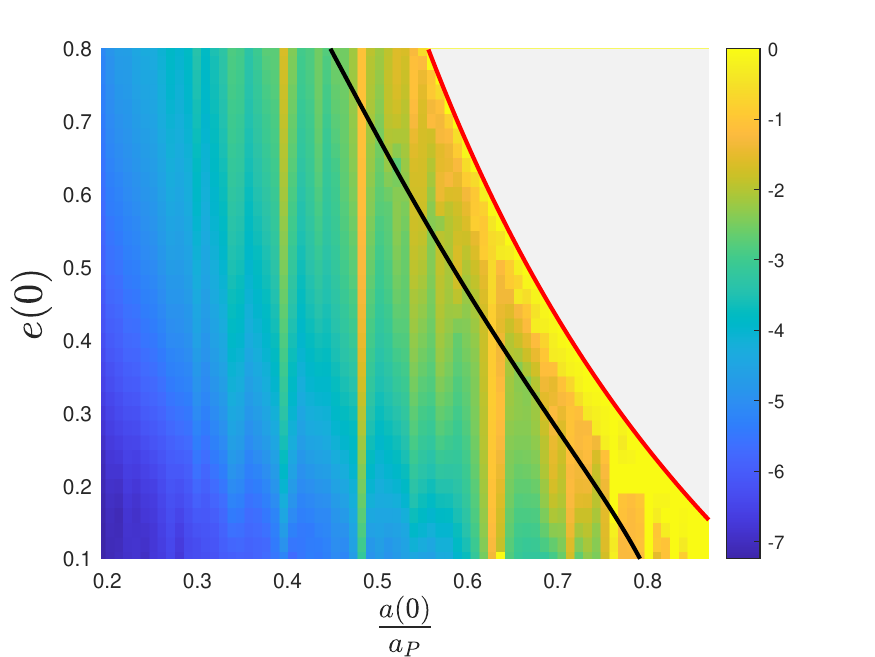}
		\subcaption{semi-major axis}
		\label{fig: smaGraph}
	\end{subfigure}
	\begin{subfigure}{0.48\textwidth}
	\centering
	\includegraphics[width=\textwidth]{numCompGraph.pdf}
	\subcaption{eccentricity}
	\label{fig: eccGraph}
\end{subfigure}
\caption{Variation of $\log_{10}{\Big(\max\limits _t\big(\big|\frac{a(t)-a_N(t)}{a_N(t)}\big|\big)\Big)}$ (left) and $\log_{10}{\Big(\max\limits _t\big(\big|\frac{e(t)-e_N(t)}{e_N(t)}\big|\big)\Big)}$ (right) with respect to the initial values of the semi-major axis and the eccentricity; $a(t)$ and $e(t)$ are computed through the normal form and the Lie transformations; $a_N(t)$ and $e_N(t)$ are computed through numerical propagation of the trajectory.}
\label{fig: smaeccGraph}
\end{figure}

\begin{figure}[h!]
	\centering
	\includegraphics[width=0.6\textwidth]{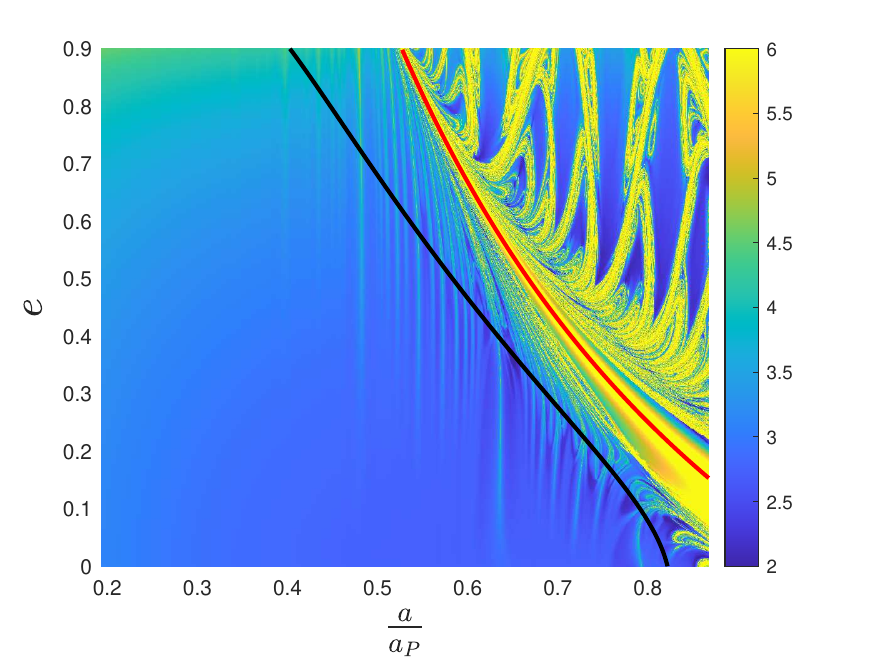}
	\caption{A numerical stability map of the domain $(a,e)$ of interest obtained through the calculation of the Fast Lyapunov Indicators (see text).}
	\label{fig: stabilitymap}
\end{figure}

As a second test, we compare the semi-analytical computation of the evolution of the orbital semi-major axis and eccentricity using the normal form with the results obtained through the numerical propagation of the particle's trajectory. 

After $s_m-s_0+1$ steps, the normalization process transforms the original canonical variables, $(\delta \Lambda^{(0)},\Gamma^{(0)},I_P^{(0)},\lambda^{(0)},\gamma^{(0)},\lambda_P^{(0)})$, into new canonical variables $(\delta \Lambda^{(s_m-s_0+1)},\Gamma^{(s_m-s_0+1)}$, 

$I_P^{(s_m-s_0+1)},\lambda^{(s_m-s_0+1)}, \gamma^{(s_m-s_0+1)},\lambda_P^{(s_m-s_0+1)})$. 
Using Hamilton's equations for the normal form Hamiltonian $Z=\sum_{s=s_0}^{s_m}Z_s$ we compute the evolution of the trajectories in the new canonical variables and back transform the result to obtain the evolution also in the original variables through the composition of the Lie series; for example
\[
\delta \Lambda^{(0)}=\exp(\{\cdot,\chi_{s_m}^{(s_m-s_0+1)}\})\exp(\{\cdot,\chi_{s_m-1}^{(s_m-s_0)}\})\dots \exp(\{\cdot,\chi_{s_0}^{(1)}\})\delta \Lambda^{(s_m-s_0+1)},
\]
with analogous formulas holding for all the other variables. Since the initial conditions of any trajectory are given in the original variables, to compute the initial conditions in the new variables, the inverse transformation must be used; for example, we have
\[
\delta \Lambda^{(s_m-s_0+1)}(0)=\exp(-\{\cdot,\chi_{s_0}^{(1)}\})\exp(-\{\cdot,\chi_{s_0+1}^{(2)}\})\dots \exp(-\{\cdot,\chi_{s_m}^{(s_m-s_0+1)}\})\delta \Lambda^{(0)}(0). \]

Since 
\[
\frac{ d \delta \Lambda^{(s_m-s_0+1)}}{dt}=\frac{d {Z}}{d\lambda^{(s_m-s_0+1)}}=0, \quad \frac{d\Gamma^{(s_m-s_0+1)}}{dt}=\frac{d {Z}}{d\gamma^{(s_m-s_0+1)}}=0
\]
$\delta \Lambda^{(s_m-s_0+1)}$ and $\Gamma^{(s_m-s_0+1)}$ are integrals of motion, while $\delta \Lambda^{(0)}$ and $\Gamma^{(0)}$ change in time. Having computed their evolution, we can obtain also the evolution of the semi-major axis and the eccentricity as
\[
a^{(0)}=\frac{(n^*{\aV}^2+\delta \Lambda^{(0)})^2}{\mathcal{GM}}, \quad e^{(0)}=\sqrt{1-\Big(1-\frac{\Gamma^{(0)}}{n^*{\aV}^2+\delta\Lambda^{(0)}}\Big)}.
\]
Note that both $a^{(s_m-s_0+1)}$ and $e^{(s_m-s_0+1)}$ are constants of motion (the  `proper' semi-major axis and  `proper' eccentricity) under the flow of $Z$ in the PC3BP. The initial conditions imposed are: $\delta \Lambda=0$, $\omega = 90^{\circ}$, $M=90^{\circ}$, $\lambda_P=0^{\circ}$. Both forward and backward propagations in time were performed for each trajectory considering an interval from $-50$ years to $50$ years. The disturbing function was computed as in the previous test; to save computational time, only $4$ steps were carried out in the normalization process. The numerical propagation was performed with \textit{MATLAB} using the function \textit{ode45}. Figure \ref{fig: smaeccGraph} shows the maximum relative errors obtained in function of the initial values of the semi-major axis and eccentricity $(a_0,e_0)$. To interpret the results, we also computed the stability map shown in Figure \ref{fig: stabilitymap}. It was obtained by computing the Fast Lyapunov Indicators (FLI) \cite{Lega1997} for orbits with the same initial conditions presented above, using a propagation time equal to $20$ orbital periods. In the Figures \ref{fig: smaGraph}, \ref{fig: eccGraph} and \ref{fig: stabilitymap}, the red and black lines are the same as described above. We observe that in the domain left to the black line, the error is generally small, except along the vertical strips corresponding to mean motion resonances and their neighbourhood, similar as in Figure \ref{fig: errorGraph}. The stability map (Figure \ref{fig: stabilitymap}) confirms these features due to  mean motion resonances. Figure \ref{fig: smaeccGraph} shows that the error of the error increases, in general, as $a$ and $e$ increase. For higher values of $a$, the error is mostly dominated by the truncation level of the multipolar expansion. Fixing $a$, the error as $e$ increases is regulated, instead, by the choice of maximum normalization order.  

From the tests performed, we can conclude that with the adopted truncation and normalization orders the method produces accurate results up to an initial semi-major axis $a_0\le
0.6a_P$; for higher values of the initial $a$, either we  accept an higher error or we must increase the order of the multipolar expansion which implies a substantial increase in computational time. On the other hand, the method gives accurate results for still high values of $e_0$, up to almost $0.7$: again more accurate results can be obtained for still higher values of $e_0$, by performing a large number of normalization steps at higher computational cost.  

\paragraph{}
Let us finally remark that, while in the case of the PCR3BP the normalization method described above can be used directly for the computation of the proper semi-major axis and proper eccentricity, the method lends itself conveniently as the starting point for the computation of proper elements also in more complex cases, e.g. when the perturbation effects of external planets are considered or in the more general spatial elliptic R3BP.  

\begin{figure}[h!]
	\begin{subfigure}{0.48\textwidth}
		\includegraphics[width=\textwidth]{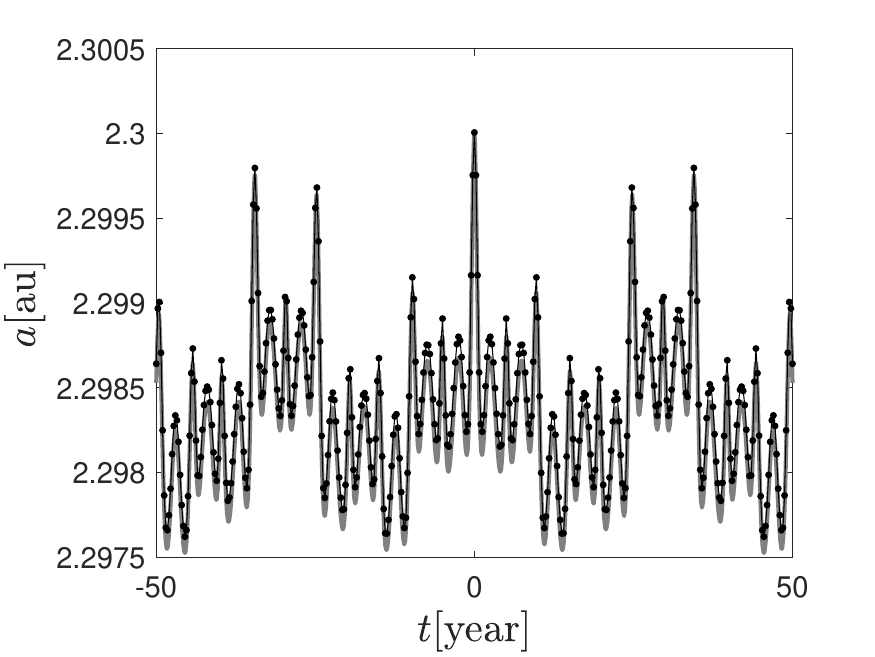}
		\subcaption{case 1 - $a(0)=2.3$ au, $e(0)=0.1$}
		\label{fig: case1}
	\end{subfigure}
	\begin{subfigure}{0.48\textwidth}
		\includegraphics[width=\textwidth]{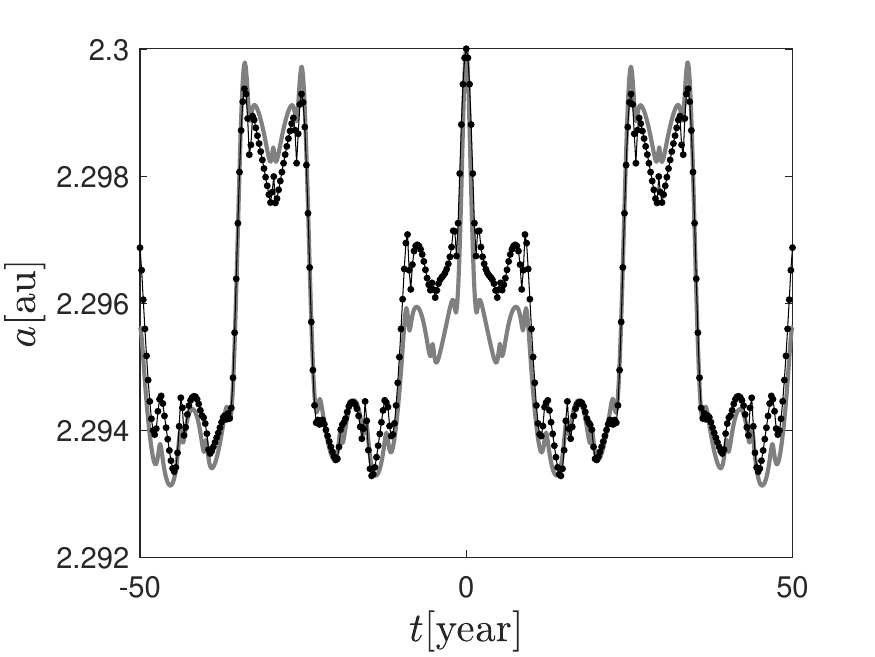}
		\subcaption{case 2 - $a(0)=2.3$ au, $e(0)=0.5$}
		\label{fig: case2}
	\end{subfigure}
	\begin{subfigure}{0.48\textwidth}
		\includegraphics[width=\textwidth]{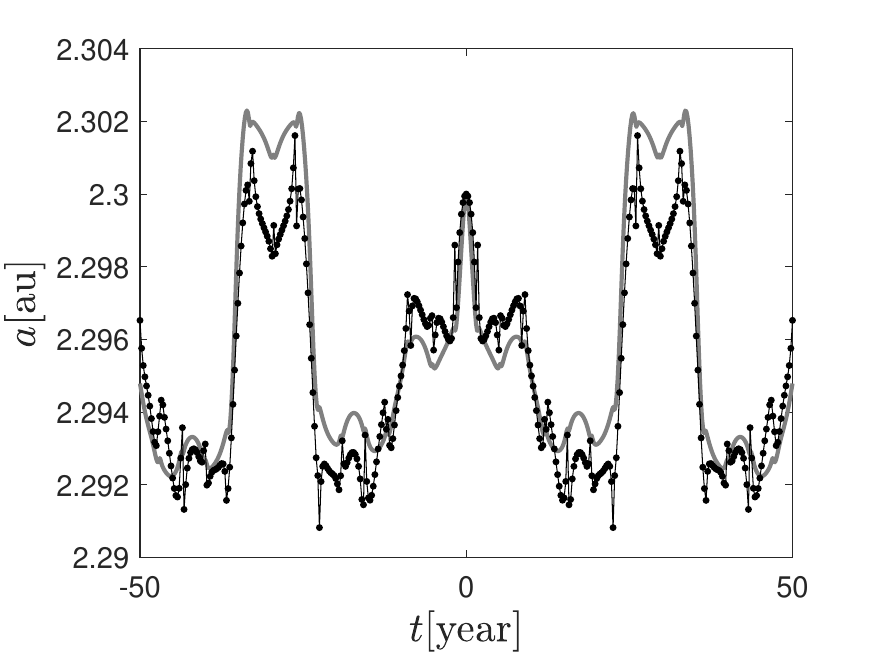}
		\subcaption{case 3 - $a(0)=2.3$ au, $e(0)=0.7$}
		\label{fig: case3}
	\end{subfigure}
	\begin{subfigure}{0.48\textwidth}
		\includegraphics[width=\textwidth]{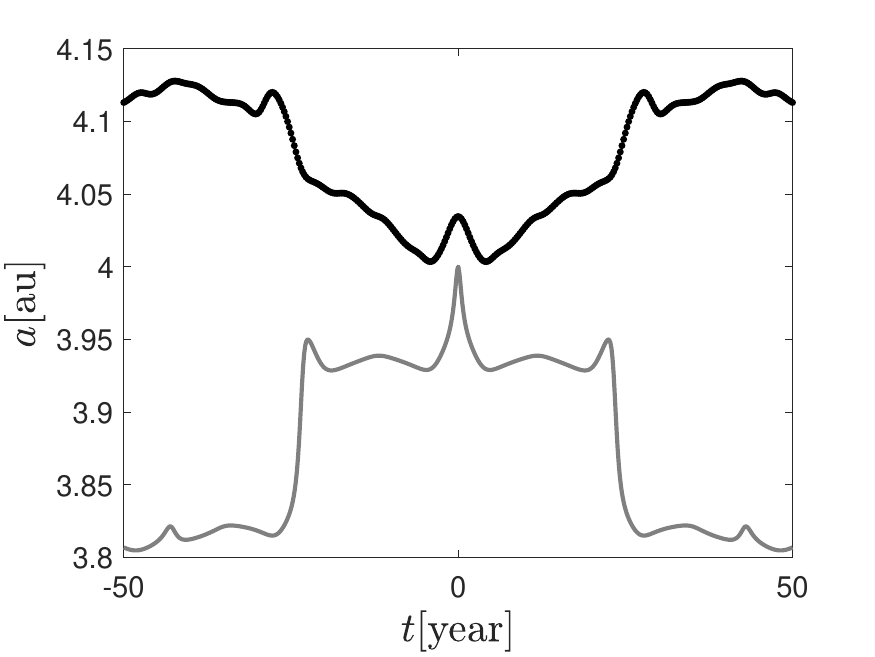}
		\subcaption{case 4 - $a(0)=4$ au, $e(0)=0.1$}
		\label{fig: case4}
	\end{subfigure}
	\caption{Comparison between the evolution of the semi-major axis computed through the normal form and the Lie transformations (black line) and that computed through a numerical propagation (grey line). The normal form and the generating functions are computed by performing $4$ normalization steps.}
	\label{fig: outcomesr3bp}
\end{figure}

\begin{figure}[h!]
	\begin{subfigure}{0.48\textwidth}
		\includegraphics[width=\textwidth]{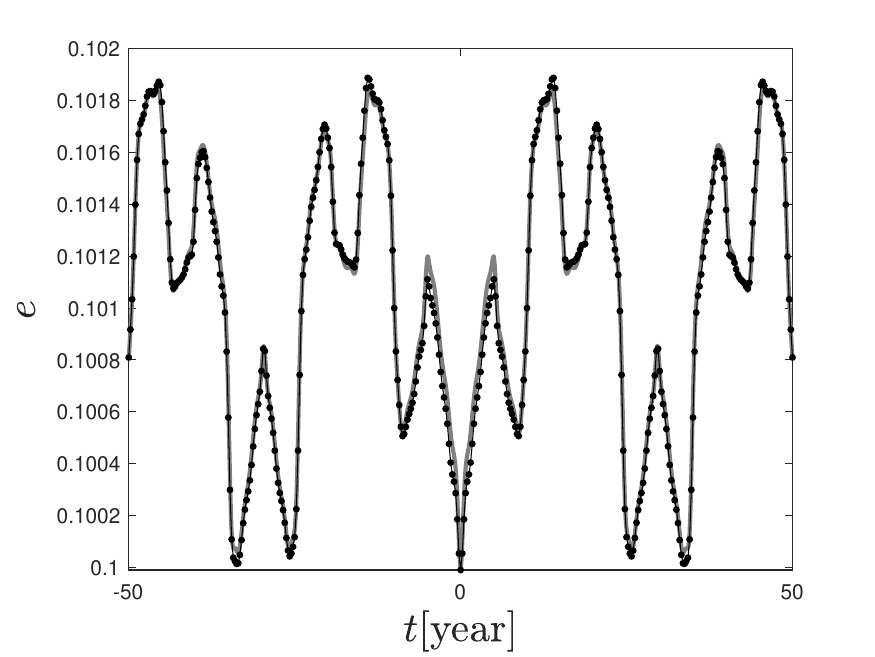}
		\subcaption{case 1 - $a(0)=2.3$ au, $e(0)=0.1$}
		\label{fig: case1ecc}
	\end{subfigure}
	\begin{subfigure}{0.48\textwidth}
		\includegraphics[width=\textwidth]{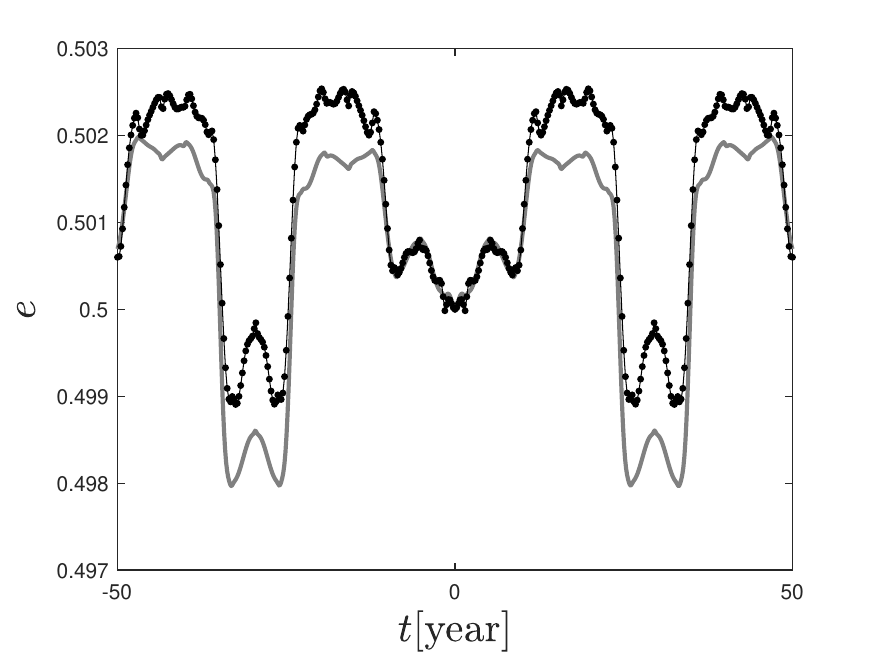}
		\subcaption{case 2 - $a(0)=2.3$ au, $e(0)=0.5$}
		\label{fig: case2ecc}
	\end{subfigure}
	\begin{subfigure}{0.48\textwidth}
		\includegraphics[width=\textwidth]{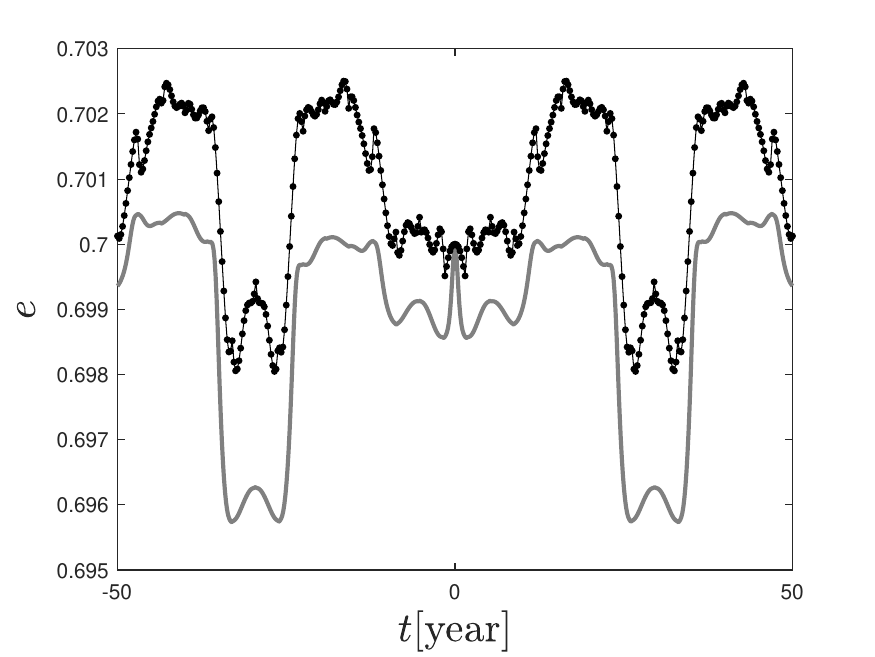}
		\subcaption{case 3 - $a(0)=2.3$ au, $e(0)=0.7$}
		\label{fig: case3ecc}
	\end{subfigure}
	\begin{subfigure}{0.49\textwidth}
		\includegraphics[width=\textwidth]{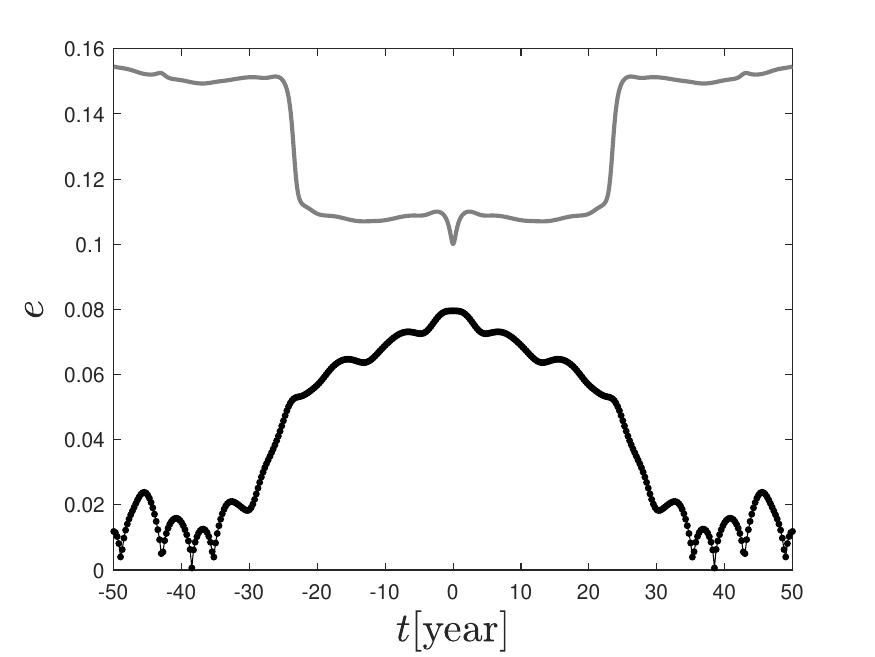}
		\subcaption{case 4 - $a(0)=4$ au, $e(0)=0.1$}
		\label{fig: case4ecc}
	\end{subfigure}
	\caption{Comparison between the evolution of the eccentricity computed through the normal form and the Lie transformations (black line) and that computed through a numerical propagation (grey line). The normal form and the generating functions are computed by performing $4$ steps of the normalization method.}
	\label{fig: outcomesr3bp_ecc}
\end{figure}

\begin{figure}[h!]
	\begin{subfigure}{0.49\textwidth}
			\centering
		\includegraphics[width=\textwidth]{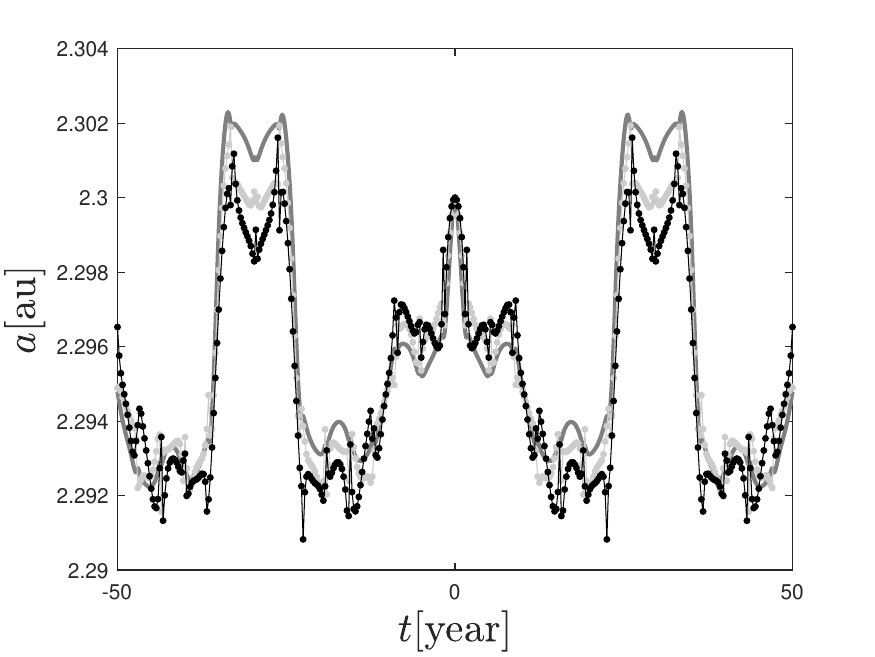}
		\subcaption{semi-major axis}
	\end{subfigure}
\begin{subfigure}{0.49\textwidth}
	\centering
	\includegraphics[width=\textwidth]{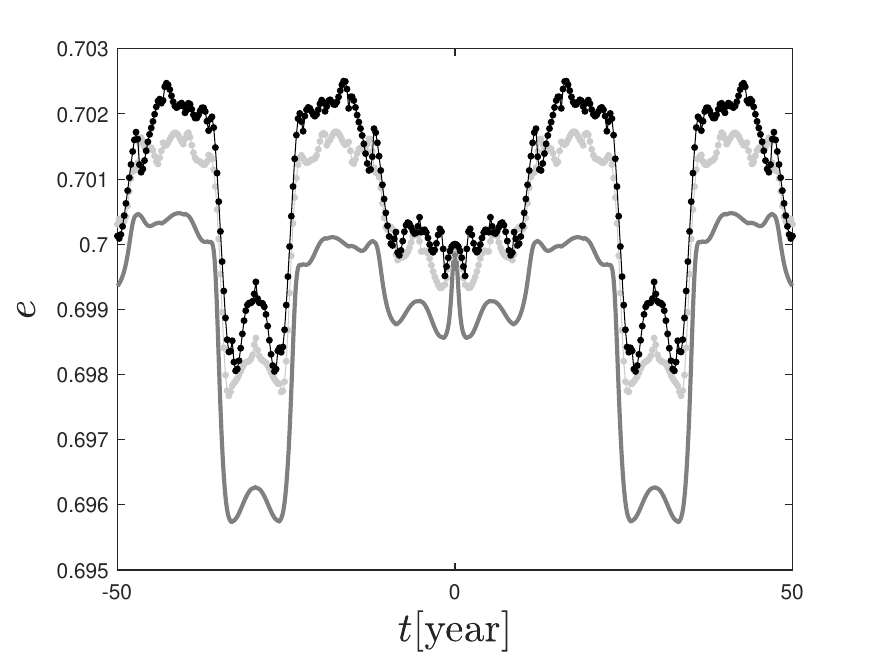}
	\subcaption{eccentricity}
\end{subfigure}
	\caption{Comparison between the evolution of the semi-major axis and the eccentricity computed through the normal form and the Lie transformations and those computed through a numerical propagation (grey line).The normal form and the generating functions are computed by performing $7$ (light grey line) and $4$ normalization steps (black line). Here, $a(0)=2.3$ au, $e(0)=0.7$.}
	\label{fig: outcomesr3bp_morestep}
\end{figure}

\section{Numerical application in the planar elliptic R3BP}
\label{section: outcomesR3BP}
We reproduced the second test described in Section \ref{section: outcomesPCR3BP} in the more general of the planar elliptic R3BP selecting few initial conditions for the particle's trajectory. In particular, we considered the following cases

\begin{itemize}
	\item case 1 : $a_0=2.3$ au  $\sim 0.44 a_P$, $e_0=0.1$;
	\item case 2 : $a_0=2.3$ au  $\sim 0.44 a_P$, $e_0=0.5$;
	\item case 3 : $a_0=2.3$ au  $\sim 0.44 a_P$, $e_0=0.7$;
	\item case 4 : $a_0=4$ au  $\sim 0.77 a_P$, $e_0=0.1$.
\end{itemize}

In all cases also impose $\delta \Lambda(0)=0$, $\omega(0) = 90^{\circ}$, $M(0)=90^{\circ}$, $i(0)=20^{\circ}$, $\Omega(0)=0^{\circ}$, $\lambda_P(0)=0^{\circ}$. The number of terms in the initial disturbing function is higher than in the case of the PCR3BP; to keep the number of operations relatively low we perform a multipolar expansion of order $5$. We carried out $4$ steps during the normalization process as in the numerical examples for the PC3BP. Figure \ref{fig: outcomesr3bp} shows the outcomes obtained for the semi-major. The method works well in the first three cases yielding a maximum relative error $\sim 10^{-4.3}$ for case 1 and $\sim 10^{-3.7}$ in the other two cases. Instead, in case 4 the method does not work properly; indeed the maximum relative error we get is $10^{-1}$. Similarly, for the eccentricity (Figure \ref{fig: outcomesr3bp_ecc}) the maximum relative error is $\sim 10^{-3.9}$ for case 1, $\sim 10^{-3.7}$ for case 2, $\sim 10^{-2.6}$ for case 3 and $\sim 10^{-0.12}$ for case 4. 

These results generally confirm the conclusions obtained for the PCR3BP. The method is able to produce accurate outcomes also for high eccentricity if the initial semi-major axis $a_0$ is sufficiently lower than $a_P$. For high values of $a_0$ the relative error depends also on the maximum order of the multipolar expansion of the original Hamiltonian. Increasing the multiple order produces a lower error, but also causes a significant increase of the computational time. Moreover, for a fixed normalization order the error increases with the orbital eccentricity: to rectify this trend it is necessary to perform a larger number of steps during the normalization process as the initial value of the eccentricity grows. For example, repeating the test for case 3, but performing $7$ steps of the normalization process, the resulting maximum relative error decreases: it reduces to $\sim 10^{-4.3}$ for the semi-major axis and to $\sim 10^{-2.8}$ for the eccentricity (see Figure \ref{fig: outcomesr3bp_morestep}).

\appendix

\section{Proof of Lemma \ref{lemmaDR}}
\label{Applemmaproof}
Consider the expression of the initial Hamiltonian (equation \eqref{H3D}) before the introduction of the book-keeping parameter. The terms obtained after the expansion of the semi-major axis $a=a^*+\delta a$, come from two parts: i) the initial Keplerian term, and ii) $\mathsf{R}$ determined through the multipole expansion of the planet's tidal potential (see \eqref{tidalpotentialme}). Moreover, the Hamiltonian contains the term $n_PI_P$. Regarding $\mathsf{R}$, we obtain the following:

\paragraph{Analysis of $\bm{\mathsf{R}}$}
From equation \eqref{tidalpotentialme}, $\mathsf{R}$ contains terms of the type  
\[
C_{j,k} \frac{r^{k-1}\cos^{j}\alpha}{r_P^{k}}, \qquad k\in\mathbb{Z}^+, k\ge 3, \qquad C_k\in\mathbb{Q}, \qquad j\in\mathbb{Z}^+.
\]
We have 
\begin{equation}
\cos \alpha = \frac{\br\cdot\br_P}{rr_P}=\frac{T_0+T_1+T_2}{r} 
\label{cosalpha}
\end{equation}
where 
\begin{equation}
\small
\begin{split}
T_0 = &  +\frac{1}{4}{a(\eta+1)}\big(\cos i + 1\big)\cos(-u+f_P-\omega-\Omega) -\frac{1}{4}{a(\eta+1)}\big(\cos i - 1\big)\cos(u+f_P+\omega-\Omega),\\
T_1 = & -\frac{1}{2}{a}e\big(\cos i +1\big)\cos(-\omega+f_P-\Omega)+\frac{1}{2}{a}e\big(\cos i -1\big)\cos(\omega+f_P-\Omega), \\
T_2 = & \frac{1}{4}\frac{ae^2(\cos i + 1)}{1+\eta}\cos(u+f_P-\omega-\Omega)-\frac{1}{4}\frac{ae^2(\cos i - 1)}{1+\eta}\cos(-u+f_P+\omega-\Omega).
\end{split}
\label{T0T1T2def}
\end{equation}
By performing the transformation \eqref{transfLemma1}, we obtain
\[
\begin{split}
T_0 = &  +\frac{1}{4}{a(\eta+1)}\big(\cos i + 1\big)\cos(-L_E+L_{T,P}) -\frac{1}{4}{a(\eta+1)}\big(\cos i - 1\big)\cos(L_E+L_{T,P}-2\Omega),\\
T_1 = & -\frac{1}{2}{a}e\big(\cos i +1\big)\cos(-\bar{\omega}+L_{T,P})+\frac{1}{2}{a}e\big(\cos i -1\big)\cos(\bar{\omega}+L_{T,P}-2\Omega), \\
T_2 = & \frac{1}{4}\frac{ae^2(\cos i + 1)}{1+\eta}\cos(L_E+L_{T,P}-2\bar{\omega})-\frac{1}{4}\frac{ae^2(\cos i - 1)}{1+\eta}\cos(-L_E+L_{T,P}+2\bar{\omega}-2\Omega).
\end{split}
\]
Then, $\cos\alpha$ fulfills the D'Alembert rules \eqref{d'al1}, \eqref{d'al2}, \eqref{d'al3}. 
Moreover, since
\[
r = a(1-e\cos u)=a(1-e \cos(L_E-\bar{\omega})),\quad \frac{1}{r_P}=\frac{1+e_P\cos f_P}{a_P\eta_P^2}=\frac{1+e_P\cos (L_{T,P}-\bar{\omega}_P)}{a_P\eta_P^2},
\]
it follows that $r$ and $1/r_P$ also fulfill the D'Alembert rules. 

The product between terms fulfilling the D'Alembert rules fulfils them as well.  
It follows that the terms coming from $\mathsf{R}$ in $\mathcal{H}$ fulfill the D'Alembert rules.

\paragraph{}
Now, all the terms of $\mathsf{R}$
are multiplied by $Q$ defined in \eqref{Qdef} (RM-reduction). $Q$ results from the expansion of the semi-major axis in 
\[
\frac{a(1-e\cos u)}{r}=\frac{a(1-e\cos(L_E-\bar{\omega}))}{r}.
\]
Thus, $Q$ fulfils the  D'Alembert rules, implying that the product $\mathsf{R}Q$ fulfils the  D'Alembert rules as well.

Finally, the Lie transformation preserves the d'Alembert rules. Thus, all the Hamiltonians computed throughout the normalization process fulfill the D'Alembert rules. 

\hfill\(\Box\)

\section{Computation of the average of the disturbing function or the generating functions $\bm{ \chi_{s_0+j-1}^{(j)}}$}
\label{AppAV}
We report some useful formulas to apply for the computation of the average of the disturbing function and of any generatring function with respect to $\lambda$. Since
\[
d \lambda = d M,
\]
the average of any trigonometric quantity over $\lambda$ coincides with the average over the mean anomaly $M$. 

From \cite{Kelly1989}, we have
\[
\frac{1}{2\pi}\int_{0}^{2\pi} \cos(ku + \nu) dM = \left\{ \begin{array}{rcl}
-\frac{e}{2}\cos \nu & \mbox{if} & |k|=1\\ 0 & \mbox{if} & |k|\neq 1 
\end{array}\right.
\]
\[
\frac{1}{2\pi}\int_{0}^{2\pi}\sin(ku + \nu) dM = \left\{ \begin{array}{rcl}
-\frac{e}{2}\sin \nu & \mbox{if} & |k|=1\\ 0 & \mbox{if} & |k|\neq 1 
\end{array}\right.
\]
\[
\frac{1}{2\pi}\int_{0}^{2\pi} (u-M) dM = 0
\]
where $k\in\mathbb{Z}$ and $\nu$ is a generic angle.

 From \cite{Kelly1989} and \cite{Kozai1962} we have the following formulas applicable whenever $f$ is used in place of $u$:
\[
\frac{1}{2\pi}\int_{0}^{2\pi} \cos(kf) dM = \frac{(-e)^{|k|}(1+|k|\eta)}{(1+\eta)^{|k|}},
\]
\[
\frac{1}{2\pi}\int_{0}^{2\pi} \sin(kf) dM = 0,
\]
\[
\frac{1}{2\pi}\int_{0}^{2\pi} \cos(kf+\nu) dM = \frac{(-e)^{|k|}(1+|k|\eta)}{(1+\eta)^{|k|}}\cos\nu,
\]
\[
\frac{1}{2\pi}\int_{0}^{2\pi} \sin(kf+\nu) dM = \frac{(-e)^{|k|}(1+|k|\eta)}{(1+\eta)^{|k|}}\sin\nu,
\]
\[
\frac{1}{2\pi}\int_{0}^{2\pi} (f-M) dM = 0.
\]

\section{Example of the normalization algorithm for $\bm{s_0>1}$}
\label{examplenormprocss}
We give below a detailed example of the proposed normalization algorithm in the generic case $s_0>1$ (i.e. $e>>{m}_P/\mathcal{M}$). We consider a toy Hamiltonian in which the initial disturbing function $R$ is given by the quadrupolar expansion (see \eqref{tidalpotentialme}): 
\[
\mathcal{H}^{(0)}=n^*\delta\Lambda+n_PI_P+R, \quad R = \epsilon^{s_0}{R}_{s_0}^{(0)}+\epsilon^{s_0+1}{R}_{s_0+1}^{(0)}+\mathcal{O}(\epsilon^{s_0+2}),
\]
where 
\begin{equation}
\begin{split}
{R}_{s_0}^{(0)}= & \frac{1}{64}\frac{\mu{\aV}^3}{r a_P^3\eta_P^6} \Big(32 -6{\etaVP}^2C_1 -6{\etaVP}^2C_2\cos(2f_P-2\Omega) -6{\etaVP}^2C_2\cos(2\uV+2\omega)\\
&-3{\etaVP}^2C_3\cos(2f_P-2\uV-2\Omega-2\omega) -3{\etaVP}^2C_4\cos(2f_P+2\uV -2\Omega+2\omega)\Big) -\frac{3}{2}\frac{\delta \Lambda^2}{\aV r}
\end{split}
\label{exR1}
\end{equation}

\begin{equation}
\begin{split}
{R}_{s_0+1}^{(0)}= & \frac{3}{128}\frac{\mu{\aV}^3}{ra_P^3\eta_P^6}\Big(4e_P\big(16-3\eta_+^2C_1\big)\cos f_P+4e\big(\eta_+(4+\eta_+)C_1-16\big)\cos u\\
&+\cos(f_P+2u+2\omega-2\Omega)+\cos(3f_P+2u+2\omega-2\Omega)\big) \\
&+e\eta_+^2C_3\cos(2f_P-3u-2\omega-2\Omega) +e\eta_+^2C_4\cos(2f_P+3u+2\omega-2\Omega)\\
&+2e\eta_+^2C_2\cos(3u+2\omega) -6e_P\eta_+^2C_2\cos(f_P-2u-2\omega)\\
&-6e_P\eta_+^2C_2\cos(f_P+2u+2\omega)-6e_P\eta_+C_2\cos(f_P-2h)\\ 
&-6e_P\eta_+C_2\cos(3f_P-2h)
+ 2e\eta_+C_2\big(8+\eta_+\big)\cos(u+2\omega)\\
&+2e\eta_+C_2\big(4+\eta_+\big)\cos(2f_P-u-2\Omega) +2e\eta_+C_2\big(4+\eta_+\big)\cos(2f_P+u-2\Omega)\\
& -3e_P\eta_+^2C_3\cos(f_P-2u-2\omega-2\Omega) -3e_P\eta_+^2C_3\cos(3f_P-2u-2\omega-2\Omega)\\
&+e\eta_+(8+\eta_+)C_3
\cos(2f_P-u-2g-2h)+e\eta_+(8+\eta_+)C_4\cos(2f_P+u+2g-2h)  \Big)
\end{split}
\label{exR2}
\end{equation}
with
\[
\etaVP=1+\etaV,\qquad C_1 = 1+\cos^2i,\qquad C_2 = 1-\cos i,\qquad C_3 = (1+\cos i)^2,\qquad C_4 = (1-\cos i)^2.
\]
Note that in ${R}_{s_0}^{(0)}$, the terms
\[
\frac{1}{2}\frac{\mu{\aV}^3}{r a_P^3\eta_P^6}, \qquad -\frac{3}{32}\frac{\mu{\aV}^3}{r a_P^3\eta_P^6}\etaVP^2C_1, \qquad -\frac{3}{2}\frac{\delta \Lambda^2}{\aV r}
\]
are of type 1, while all the others are of the type 2 (see section \ref{section: normalization} for the definition of type 1 and 2); ${R}_{s_0+1}^{(0)}$ has only terms of type 2.

At the first step of the iterative method, the goal is to determine the function $\chi_{s_0}^{(1)}$ to normalize $\mathcal{H}^{(0)}$ up to the order $s_0$ in $\epsilon$: then, the term to be normalized is $R^{(0)}_{s_0}$. In view of equation \eqref{exR1}, we obtain the following generating function  $\chi_{s_0}^{(1)}$:
\[
\begin{split}
\chi_{s_0}^{(1)}= & \frac{1}{64}\frac{\mu{\aV}^2}{a_P^3\eta_P^6} \Bigg(\frac{32}{n^*}\phi\epsilon -\frac{6\eta_{+}^2C_1}{n^*}\phi\epsilon -\frac{6\eta_{+}^2C_2\sin(2f_P-2\Omega)}{2n_P} -\frac{6\eta_{+}^2C_2\sin(2u+2\omega)}{2n^*}\\
&-\frac{3\eta_{+}^2C_3\sin(2f_P-2u-2\Omega-2\omega)}{-2n^*+2n_P} -\frac{3\eta_{+}^2C_4\sin(2f_P+2u-2\Omega+2\omega)}{2n^*+2n_P}\Bigg)\\
& -\frac{3}{2} \frac{\delta \Lambda^2}{{\aV}^2}\phi\epsilon. 
\end{split}
\]  
The homological equation (equations \eqref{homeq} and \eqref{chieq}) contains the normal form terms
\begin{equation}
Z_{s_0}= \frac{1}{32}\frac{\mu{\aV}^2}{a_P^3\eta_P^6} \Bigg(16 -3\eta_{+}^2C_1\Bigg) -\frac{3}{2}\frac{\delta \Lambda^2}{\aV}
\label{dafpmp}
\end{equation}
and the remainder terms (of order lower than $2s_0$) given by 
\[
\begin{split}
\hat{{R}}_{(s_0+1)}^{(1)}=& \frac{3}{128}\frac{\mu{\aV}^3}{ra_P^3\eta_P^9}\Bigg(2\frac{n_P}{n_P-n^*}e_P\eta_+^2C_3\Big(\cos(-2u+3f_P-2\omega-2\Omega)\\ &+\cos(-2u+f_P-2\omega-2\Omega)\Big)
+2\frac{n_P}{n_P+n^*}e_P\eta_+^2C_4\Big(\cos(2u+3f_P+2\omega-2\Omega)\\
&+\cos(2u+f_P+2\omega-2\Omega)\Big)
-\frac{n_P}{n_P-n^*}\eta_P^3e\eta_+^2C_3\Big(\cos(-3u+2f_P-2\omega-2\Omega)\\
&+\cos(-u+2f_P-2\omega-2\Omega)\Big)
-\frac{n_P}{n_P+n^*}\eta_P^3e\eta_+^2C_3\Big(\cos(3u+2f_P+2\omega-2\Omega)\\
& +\cos(u+2f_P+2\omega-2\Omega)\Big)
+ 4e_P\eta_+^2C_2\Big(\cos(3f_P-2\Omega)+\cos(f_P-2\Omega)\Big) \\
&-2\eta_P^3e\eta_+^2C_2\Big(\cos(-u+2f_P-2\Omega)
+\cos(u+2f_P-2\Omega)\Big) \Bigg),\end{split}\]
\[
\begin{split}
\hat{{R}}_{(s_0+2)}^{(1)}= & \frac{3}{256}\frac{\mu{\aV}^3}{ra_P^3\eta_P^9}\Bigg( \frac{n_P}{n_P-n^*}e_P^2\eta_+^2C_3\Big(\cos(2u+2\omega+2\Omega)+\cos(-2u+4f_P-2\omega-2\Omega)\Big)\\
&\frac{n_P}{n_P+n^*}e_P^2\eta_+^2C_4\Big(\cos(2u+4f_P+2\omega-2\Omega)+\cos(2u+4f_P+2\omega-2\Omega)\Big)\\
& -2\frac{n_P}{n_P-n^*}e_Pe\eta_+^2C_3\Big(\cos(-3u+3f_P-2\omega-2\Omega)+\cos(-u+3f_P-2\omega-2\Omega)\\
&+\cos(-3u+f_P-2\omega-2\Omega)+\cos(-u+f_P-2\omega-2\Omega)\Big)\\
&-2\frac{n_P}{n_P+n^*}e_Pe\eta_+^2C_4\Big(\cos(3u+3f_P+2\omega-2\Omega)+\cos(u+3f_P+2\omega-2\Omega)\\
&+\cos(3u+f_P+2\omega-2\Omega)+\cos(u+f_P+2\omega-2\Omega)\Big)\\
& +2\frac{n_P}{n_P-n^*}(2+e_P^2-2\eta_P^3)\eta_+^2C_3\cos(-2u+2f_P-2\omega-2\Omega)\\
&+2\frac{n_P}{n_P+n^*}(2+e_P^2-2\eta_P^3)\eta_+^2C_3\cos(2u+2f_P+2\omega-2\Omega)\\
&+2e_P^2\eta_+^2C_2cos(2\Omega) +2e_P^2\eta_+^2C_2\cos(4f_P-2\Omega)\\
\end{split}
\]
\[
\begin{split}
&+4(2+e_P^2-2\eta_P^3)\eta_+^2C_2\cos(2f_P-2\Omega)-4e_Pe\eta_+^2C_2\Big(\cos(-u+3f_P-2\Omega)\\
&+\cos(u+3f_P-2\Omega)+ \cos(-u+f_P-2\Omega)+\cos(u+f_P-2\Omega)\Big)\Bigg),\\
\end{split}
\]
\[
\begin{split}
\hat{{R}}_{(s_0+3)}^{(1)}= & -\frac{3}{512}\frac{\mu{\aV}^3}{ra_P^3\eta_P^9}\Bigg(\frac{n_P}{n_P-n^*}e_P^2e\eta_+^2C_3\Big(\cos(u+2\omega+2\Omega)+\cos(3u+2\omega+2\Omega)\\
&+\cos(-3u+4f_P-2\omega-2\Omega)+\cos(-u+4f_P-2\omega-2\Omega)\Big)\\
&+2e_P^2e\eta_+^2C_2\Big(\cos(-u+2\Omega)+\cos(u+2\Omega)+\cos(-u+4f_P-2\Omega)\\ &+\cos(u+4f_P-2\Omega)\Big) 
+4(2+e_P^2-2\eta_P^3)e\eta_+^2C_2\Big(\cos(-u+2f_P-2\Omega)\\ &+\cos(u+2f_P-2\Omega)\Big)
+\frac{n_P}{n_P+n^*}e_P^2e\eta_+^2C_4\Big(\cos(u+2\omega-2\Omega)+\cos(3u+2\omega-2\Omega)\\
\end{split}
\]
\[
\begin{split}
&+\cos(3u+4f_P+2\omega-2\Omega)+\cos(u+4f_P+2\omega-2\Omega)\Big)\\
& +2\frac{n_P}{n_P-n^*}(2+e_P^2-2\eta_P^3)e\eta_+^2C_3\Big(\cos(-3u+2f_P-2\omega-2\Omega)\\ &+\cos(-u+2f_P-2\omega-2\Omega)\Big) +2\frac{n_P}{n_P+n^*}(2+e_P^2-2\eta_P^3)e\eta_+^2C_4\Big(\cos(3u+2f_P+2\omega-2\Omega)\\ &+\cos(u+2f_P+2\omega-2\Omega)\Big)\Bigg).\\
\end{split}
\]
The new Hamiltonian is
\[
\mathcal{H}^{(1)}=\exp(\mathcal{L}_{\chi_{s_0}^{(1)}})\mathcal{H}^{(0)}=Z_0+Z_{s_0}+\sum_{s=s_0+1}^{s_m}\epsilon^s{{R}}_{s}^{(1)};
\]
For all $ s\in [s_0+1,s_m]$ ${{R}}_{s}^{(1)}$ contains the following contributions : i) ${{R}}_{(s)}^{(0)}$; ii) the remainder of the homological equation; iii) the terms generated by the Lie transformation. From the previous considerations, we have that these last ones have book-keeping order equal to or larger than $2s_0$.

At the second step, the procedure is repeated with the goal of determining the generating function $\chi_{s_0+1}^{(2)}$ to normalize $\mathcal{H}^{(1)}$ up to the order $s_0+1$ in $\epsilon$. The term to normalize is
\[
R_{s_0+1}^{(1)}=R_{s_0+1}^{(0)}+\hat{R}_{s_0+1}^{(1)},
\]
with ${R}_{s_0+1}^{(0)}$ given in \eqref{exR2} and $\hat{R}_{s_0+1}^{(1)}$ equal to the order $s_0+1$ (in $\epsilon$) term of the remainder computed at the previous step. The generating function $\chi_{s_0+1}^{(2)}$ then is computed as
\begin{equation*}
\begin{split}
\chi_{s_0+1}^{(2)}= & \frac{3}{128}\frac{\mu{\aV}^2}{ra_P^3\eta_P^6}\Bigg(4e_P\big(16-3\eta_+^2C_1\big)\frac{\sin f_P}{n_P}+4e\big(\eta_+(4+\eta_+)C_1-16\big)\frac{\sin u}{n^*}\\
&+\frac{\sin(f_P+2u+2\omega-2\Omega)}{2n^*+n_P}+\frac{\sin(3f_P+2u+2\omega-2\Omega)}{2n^*+3n_P}\big) \\
&+e\eta_+^2C_3\frac{\sin(2f_P-3u-2\omega-2\Omega)}{2n_P-3n^*} +e\eta_+^2C_4\frac{\sin(2f_P+3u+2\omega-2\Omega)}{3n^*+2n_P}\\
&+2e\eta_+^2C_2\frac{\sin(3u+2\omega)}{3n^*} -6e_P\eta_+^2C_2\frac{\sin(f_P-2u-2\omega)}{n_P-2n^*}\\
&-6e_P\eta_+^2C_2\frac{\sin(f_P+2u+2\omega)}{n_P+2n^*}-6e_P\eta_+C_2\frac{\sin(f_P-2h)}{n_P}\\
&-6e_P\eta_+C_2\frac{\sin(3f_P-2h)}{3n_P}
+ 2e\eta_+C_2\big(8+\eta_+\big)\frac{\sin(u+2\omega)}{n^*}\\
&+2e\eta_+C_2\big(4+\eta_+\big)\frac{\sin(2f_P-u-2\Omega)}{2n_P-n^*} +2e\eta_+C_2\big(4+\eta_+\big)\frac{\sin(2f_P+u-2\Omega)}{2n_P+n^*}\\
& -3e_P\eta_+^2C_3\frac{\sin(f_P-2u-2\omega-2\Omega)}{n_P-2n^*} -3e_P\eta_+^2C_3\frac{\sin(3f_P-2u-2\omega-2\Omega)}{3n_P-2n^*}\\
&+e\eta_+(8+\eta_+)C_3
\frac{\sin(2f_P-u-2g-2h)}{2n_P-n^*}+e\eta_+(8+\eta_+)C_4\frac{\sin(2f_P+u+2g-2h)}{2n_P+n^*} \\
& + 2\frac{n_P}{n_P-n^*}e_P\eta_+^2C_3\Big(\frac{\sin(-2u+3f_P-2\omega-2\Omega)}{3n_P-n^*}\\
 &+\frac{\sin(-2u+f_P-2\omega-2\Omega)}{n_P-2n^*}\Big)
+2\frac{n_P}{n_P+n^*}e_P\eta_+^2C_4\Big(\frac{\sin(2u+3f_P+2\omega-2\Omega)}{2n^*+3n_P}\\
&+\frac{\sin(2u+f_P+2\omega-2\Omega)}{2n^*+n_P}\Big)
-\frac{n_P}{n_P-n^*}\eta_P^3e\eta_+^2C_3\Big(\frac{\sin(-3u+2f_P-2\omega-2\Omega)}{2n_P-3n^*}\\
\end{split}
\end{equation*}
\begin{equation*}
\begin{split}
&+\frac{\sin(-u+2f_P-2\omega-2\Omega)}{2n_P-n^*}\Big)
-\frac{n_P}{n_P+n^*}\eta_P^3e\eta_+^2C_3\Big(\frac{\sin(3u+2f_P+2\omega-2\Omega)}{3n^*+2n_P}\\
& +\frac{\sin(u+2f_P+2\omega-2\Omega)}{n^*+2n_P}\Big)
+ 4e_P\eta_+^2C_2\Big(\frac{\sin(3f_P-2\Omega)}{3n_P}+\frac{\sin(f_P-2\Omega)}{n_P}\Big) \\
&-2\eta_P^3e\eta_+^2C_2\Big(\frac{\sin(-u+2f_P-2\Omega)}{2n_P-n^*}
+\frac{\sin(u+2f_P-2\Omega)}{n^*+2n_P}\Big) \Bigg),\\
\end{split}
\end{equation*}
while we also have
\[
Z_{s_0+1}=0.
\]
Subsequent steps can be computed using analogous formulas. 

\section{Proof of Proposition \ref{PropositionImportant}}
\label{AppPropProof}

\paragraph{i)}
Consider the $j$-th normalization step in the case $s_0=1$.  From Remark \ref{remarkCriticalTerms} it follows that the Poisson bracket $\{Z_{s_0},\chi_j^{(j)}\}$ yields terms of order $s_0+j-1=j$ through formula \eqref{c_orderm1}, by taking $A_1=Z_1$ and $A_2=\chi_{j}^{(j)}$.  However, since $Z_1$ is a normal form term, it does not depend on $u$ and $r$. Moreover, since it has book-keeping order equal to $s_0=1$, it does not depend explicitly on $e$ (as a consequence of the adopted book-keeping rules). By Lemma \ref{lemmaDR}, we conclude that it does not depend on $\omega$. It follows that the term of book-keeping order $j$ coming from  $\{Z_{s_0},\chi_j^{(j)}\}$ is equal to zero. 

\paragraph{ii,iii)}
To show the second and third points of Proposition \ref{PropositionImportant}, we use the following lemma:
\begin{lemma}
	All terms in $R_2^{(1)}$ are of the following two types:  
	\begin{itemize}
		\item terms of type A, not depending on the eccentricity $e$ and $\phi$; 
		\item terms of type B, linearly depending on the eccentricity $e$ (and not depending on $\phi$). 
	\end{itemize} 
	All terms of type A are of the form \eqref{A1k1k3cos}. All terms of type B depend on the eccentric anomaly $u$, the true anomaly $f_P$, or both. 
	\label{lemmacentrale}
\end{lemma}
Since at the second step $R_2^{(1)}$ is the term to be normalized, it follows that all terms in the generating function $\chi_2^{(2)}$ also satisfy Lemma \ref{lemmacentrale}.

\paragraph{}
By Lemma \ref{lemmacentrale}, the normal form term $Z_2$ obtained by normalizing $R_2^{(1)}$ does not depend explicitly on $e$ and $\phi$. Moreover, by definition, any normal form terms cannot depend on $u$. By Lemma \ref{lemmaDR} it cannot depend on $\omega$. Thus 
\[
\frac{\partial Z_{2}}{\partial e}=0, \quad \frac{\partial Z_{2}}{\partial \omega}=0, \quad \frac{\partial Z_{2}}{\partial u}=0, \quad \frac{\partial Z_{2}}{\partial \phi}=0.
\] 
 It follows from Remark \ref{remarkOrderM2} (considering $A_1=Z_{2}$) that $\{Z_{2},\chi_{j}^{(j)}\}\sim \epsilon^{j+1}$ $\forall j>2$. This concludes the demonstration of point {(ii)} of Proposition \ref{PropositionImportant}.

\paragraph{}
To demonstrate point (iii) of Proposition \ref{PropositionImportant}, we finally need to prove that 
\[
\{Z_s,\chi_{2}^{(2,\rm bis)}\}\sim \epsilon^3, \quad \mbox{with} \quad s=1,2, 
\]
\[
\{R_2^{(2)},\chi_{2}^{(2,\rm bis)}\}\sim \epsilon^3. 
\]

Let $R_{2}^{(2)}$ be the terms of book-keeping order $2$ coming from $\{R_{2}^{(1)},\chi_2^{(2)}\}$. Since all terms in $R_{2}^{(1)}$ and in $\chi_2^{(2)}$ are of type A or B, from Remark \ref{k1k3sameremark} it follows that $R_{2}^{(2)}$ is produced by a Poisson bracket of the form
\[
\{T_{R_{2}^{(1)},B},T_{\chi_{2}^{(2)},B}\}
\]
where $T_{R_{2}^{(1)},B}$ and $T_{\chi_{2}^{(2)},B}$ are any terms of type B contained in $R_{2}^{(1)}$ and $\chi_{2}^{(2)}$ respectively. More specifically, from Remark \ref{remarkOrderM2} we have that $R_{2}^{(2)}$ is produced by the part of $\{T_{R_{2}^{(1)},B},T_{\chi_{2}^{(2)},B}\}$ given in \eqref{c_orderm2} with $A_1=T_{R_{2}^{(1)},B}$ and $A_2 =T_{\chi_{2}^{(2)},B}$. However, since $T_{R_{2}^{(1)},B}$ and $T_{\chi_{2}^{(2)},B}$ are linear functions of $e$ (not depending on $\phi$), applying the formulas listed in section \ref{sectionPB} to compute $\{T_{R_{2}^{(1)},B},T_{\chi_{2}^{(2)},B}\}$, we obtain that $R_{2}^{(2)}$ does not depend on $e$. Thus, by Lemma \ref{lemmaDR} $R_{2}^{(2)}$ is made by terms of the form \eqref{A1k1k3cos}. It follows that:
\begin{itemize}
	\item the generating function $\chi_{2}^{(2,\rm bis)}$ contains terms of the form \eqref{A1k1k3sin}; 
	\item from Remark \ref{k1k3sameremark}, we have $\{Z_2,\chi_{2}^{(2,\rm bis)}\}\sim \epsilon^3$ and $\{R_{2}^{(2)},\chi_{2}^{(2,\rm bis)}\}\sim \epsilon^3$.
	\item the normal form term $Z_{2,\rm bis}$ does not depend on $e$ or $\omega$. Moreover, by definition, it does not depend on $u$ and $\phi$. Thus, in view of Remark \ref{remarkOrderM2} we have that $\{Z_{2,\rm bis},\chi_{j}^{(j)}\}\sim \epsilon^{j+1}$, $\forall j\ge3$ (as also shown for $Z_2$). 
\end{itemize}

From the first point of the Proposition, we have $\{Z_1,\chi_{2}^{(2,\rm bis)}\}\sim \epsilon^3$.
This concludes the demonstration of point (iii) of Proposition \ref{PropositionImportant}.

\hfill\(\Box\)

\paragraph{Proof of Lemma \ref{lemmacentrale}}
At the second step of the normalization process, the remainder term to normalize is $R_2^{(1)}$, which is of book-keeping order $2$. Because of the book-keeping rules, $R_2^{(1)}$ contains terms which can only depend on the following factors (powers or products) of small parameters: 
	\begin{itemize}
		\item factor 1: $\mu^2$
		\item factor 2: $\delta \Lambda^3$
		\item factor 3: $\mu \delta \Lambda$
		\item factor 4: $\mu e_P$
		\item factor 5: $\delta \Lambda e_P$
		\item factor 6: $\mu e$
		\item factor 7: $\delta \Lambda e$
		\item factor 8: $\mu \phi$
		\item factor 9:  $\delta \Lambda \phi$
	\end{itemize}
All terms in $R_2^{(1)}$ containing one of the first five factors are of type A; all terms in $R_2^{(1)}$ containing the factors 6 and 7 are of type B. Considering that $\phi=e\sin u$, by substitution we have that also all terms containing factors 8 and 9 are of type B. 

From Lemma \ref{lemmaDR}, it follows that the terms of type A are of the form \eqref{A1k1k3cos}. 

To show that all the terms of type B depend on $u$, $f_P$, or both, we need to examine the following parts of $R_2^{(1)}$:
\begin{enumerate}[i)]
	\item  $R_{2}^{(0)}$, stemming from terms of book-keeping order $2$ in the original Hamiltonian,
	\item the remainder terms produced by the homological equation and the Lie transformation at the first step of the normalization process. 
\end{enumerate}

\begin{itemize}
	\item[1)] Analysis of $R_2^{(0)}$
\end{itemize}

Let us consider, first, the initial Hamiltonian \eqref{H3D}. 
In view of the expressions for $\cos \alpha$ in \eqref{cosalpha}, $r$ in \eqref{ratior} and $r_P$ in \eqref{rPdef}, we obtain that $\mathsf{R}$ (see \eqref{tidalpotentialme}) contains terms of the form
\[
\mu D_k\frac{r^{k-1}}{r_P^{k}}=\frac{\mu D_k a^{k-1}}{a_P^k\eta_P^{2k}}\sum_{l=0}^{k}\sum_{q=0}^{k-1}\Big(\begin{array}{c} k \\ l \end{array}\Big)\Big(\begin{array}{c} k-1 \\ q \end{array}\Big)(-1)^q\epsilon^{1+q+l}e^qe_P^l\cos^q u\cos^l f_P
\]
or
\[
\begin{split}
\mu C_{k,j}\frac{r^{k-1}\cos^{j}\alpha}{r_P^{k}}= \frac{\mu C_{k,j}a^{k-2}}{a_P^k\eta_P^{2k}}\sum_{l=0}^{k}\sum_{q=0}^{k-1-j}\sum_{m=0}^{j}\sum_{z=0}^{m}&\Big(\begin{array}{c} k \\ l \end{array}\Big)\Big(\begin{array}{c} k-1-j \\ q \end{array}\Big)\Big(\begin{array}{c} j \\ m \end{array}\Big)\Big(\begin{array}{c} m \\ z \end{array}\Big)(-1)^q\\ &\epsilon^{1+q+l+m+z} e^q e_P^l\cos^q u\cos^l f_PT_0^{j-m}T_1^{m-z}T_2^{z}.
\end{split}
\]
with $D_k,C_{k,j}\in\mathbb{Q}$, $k,j\in\mathbb{Z}^+$, $k\ge 3$, $j\ge1$. 

After expanding the semi-major axis (equation \eqref{smaexp}), the terms of book-keeping order $1$ of $\mathsf{R}$ have one of the forms
\begin{equation*}
\frac{\mu}{a_P^k\eta_P^{2k}}\bar{D}_k{\aV}^{k-1}, \qquad \frac{\mu}{a_P^k\eta_P^{2k}}\bar{C}_{k,j}{\aV}^{k-2}T_0^{j}
\end{equation*}
and those of book-keeping order $2$ have one of the forms
\[
\frac{\mu}{a_P^k\eta_P^{2k}}\bar{D}_k{\aV}^{k-1}e\cos u, \qquad 
\frac{\mu}{a_P^k\eta_P^{2k}}\bar{C}_{k,j}{\aV}^{k-2}e\cos uT_0^{j}, \qquad  \frac{\mu}{a_P^k\eta_P^{2k}}{C}_{k,j}{\aV}^{k-2}T_0^{j-1}T_1,
\]
\[
\frac{\mu}{a_P^k\eta_P^{2k}}\bar{D}_k{a}^{k-1}e_P\cos f_P, \qquad 
\frac{\mu}{a_P^k\eta_P^{2k}}\bar{C}_{k,j}{a}^{k-2}e_P\cos f_PT_0^{j}.
\]

Then, considering also the terms from the Keplerian part and performing the product by $Q$, $R_1^{(0)}$ contains terms of the form
\[
\frac{\aV}{r}\frac{\mu}{a_P^k\eta_P^{2k}}\bar{D}_k{\aV}^{k-1}, \qquad \frac{\aV}{r}\frac{\mu}{a_P^k\eta_P^{2k}}\bar{C}_{k,j}{\aV}^{k-2}T_0^{j}, \qquad -\frac{\aV}{r}\frac{3}{2}\frac{\delta\Lambda^2}{{\aV}^2},
\]
with $\bar{D}_k, \bar{C}_{k,j}\in \mathbb{Q}$; we also have that $R_2^{(0)}$ contains terms linearly depending on $e$, in one of the forms 
\begin{equation}
\begin{split}
& \frac{\aV}{r}\frac{\mu}{a_P^k\eta_P^{2k}}\bar{D}_k{\aV}^{k-1}e\cos u, \qquad 
\frac{\aV}{r}\frac{\mu}{a_P^k\eta_P^{2k}}\bar{C}_{k,j}{\aV}^{k-2}e\cos uT_0^{j}, \qquad  \frac{\aV}{r}\frac{\mu}{a_P^k\eta_P^{2k}}\bar{C}_{k,j}{\aV}^{k-2}T_0^{j-1}T_1, \\
& -\frac{\aV}{r}\frac{\mu}{a_P^k\eta_P^{2k}}\bar{D}_k{\aV}^{k-1}e\cos u, \qquad -\frac{\aV}{r}\frac{\mu}{a_P^k\eta_P^{2k}}\bar{C}_{k,j}{\aV}^{k-2}T_0^{j}e\cos u, \qquad \frac{\aV}{r}\frac{3}{2}\frac{\delta\Lambda^2}{{\aV}^2}e\cos u,
\end{split}
\label{termELinear}
\end{equation}
as well as terms of the form 
\[
\frac{\aV}{r}\frac{\mu}{a_P^k\eta_P^{2k}}\bar{D}_k{a}^{k-1}e_P\cos f_P, \qquad 
\frac{\aV}{r}\frac{\mu}{a_P^k\eta_P^{2k}}\bar{C}_{k,j}{a}^{k-2}e_P\cos f_PT_0^{j}, \qquad \frac{\aV}{r}\mathcal{O}(\delta \Lambda^3).
\]

We conclude that $R_1^{(0)}$ contains only terms of type 1 or type 2 (see section \ref{section: normalization}) which do not depend on the eccentricity and are consequently of the form \eqref{A1k1k3cos}. Moreover, all the terms of $R_2^{(0)}$ which depend linearly on the eccentricity (equation \eqref{termELinear}) also necessarily depend on $u$, $f_P$, or both, as we readily to infer by considering the expressions of $T_0$ and $T_1$ in \eqref{T0T1T2def}.

\begin{itemize}
	\item[2)] Analysis of $R_2^{(1)}$
\end{itemize}
$R_2^{(1)}$ is composed by three contributes:
\[
R_{2}^{(1)}=R_{2}^{(0)}+\hat{R}_{2}^{(1)}+\hat{{R}}_{LT,2}^{(1)},
\]
where we denote by $\hat{R}_{2}^{(1)}$ and $\hat{{R}}_{LT,2}^{(1)}$ the parts of the remainder respectively coming from the homological equation and the Lie transformation at the first step of the normalization process. 
It has been already shown that the terms of $R_{2}^{(0)}$ linearly depending on $e$ depend also on $u$, $f_P$, or both. Furthermore, the remainder term ${R}_{1}^{(0)}$ (normalized at the first step) contains only terms of type 1 and 2 of the form 
\[
\epsilon\frac{\aV}{r}f(i,\eta,\Omega), \qquad 
\epsilon\frac{\aV}{r} \hat{f}_{\bm{k}}(i,\eta)\cos(k_1 u + k_2 f_P + k_3\omega+k_4\Omega), \quad k_1=k_3 
\]
with either $k_1\ge1$ or $k_2\ge1$. It follows that the residual of the homological equation of the first step is produced only by terms of type 2. To normalize these last terms, $\chi_{1}^{(1)}$ has to acquire terms of the form
\[
\frac{\epsilon}{k_1n^*+k_2n_P} \hat{f}_{\bm{k}}(i,\eta)\sin(k_1 u + k_2 f_P + k_3\omega+k_4\Omega).
\]
Applying, now, the homological equation \eqref{homeq} for the function $\chi_{1}^{(1)}$, we obtain that the residual of the homological equation yields remainder terms of book-keeping order $2$, which are of the form
\begin{equation*}
\begin{split}
\hat{{R}}^{(1)}_{{2},\bm{k}}=&-\hat{f}_{\bm{k}}(i,\eta) \frac{n_Pk_2}{n^*k_1+n_Pk_2}\frac{1}{8\eta_P^3}\frac{\aV}{r}\Bigg(-4e\eta_P^3\big(\cos((k_1-1)u+k_2f_P+k_3\omega+k_4\Omega)\\
&+\cos((k_1+1)u+k_2f_P+k_3\omega+k_4\Omega)\big)+8e_P\big(\cos(k_1u+(k_2-1)f_P+k_3\omega+k_4\Omega)\\
&+\cos(k_1u+(k_2+1)f_P+k_3\omega+k_4\Omega)\big)\Bigg).\\
\end{split}
\end{equation*}
All terms above with $k_2=0$ are equal to zero; if $|k_2|\ge1$ we find that all terms linearly depending on $e$ necessarily also depend on $f_P$.

Now, when the Lie transformation is performed at the first step, the remainder terms of book-keeping order $2$ are produced by the Poisson bracket
\[
\{R_1^{(0)},\chi_1^{(1)}\}+\{R_2^{(0)},\chi_1^{(1)}\}
\]
as we can deduce from Proposition \ref{pboutcomesorder}. 
Neither $R_1^{(0)}$ or $\chi_1^{(1)}$ depends on the eccentricity. Then, using the formulas of section \ref{sectionPB}, it is easy to verify that the Poisson bracket $\{R_1^{(0)},\chi_1^{(1)}\}$ does not produce any term linearly depending on $e$. The  Poisson bracket $\{R_2^{(0)},\chi_1^{(1)}\}$ produces remainder terms of book-keeping order $2$ through the contributions
\[
\begin{split}
& \Big(\frac{\partial  R_2^{(0)}}{\partial u}\frac{\partial u}{\partial \gamma}+\frac{\partial  R_2^{(0)}}{\partial \omega}\frac{\partial \omega}{\partial \gamma}+\frac{\partial  R_2^{(0)}}{\partial \phi}\frac{\partial \phi}{\partial \gamma}\Big)\Big(\frac{\partial \chi_1^{(1)}}{\partial u}\frac{\partial u}{\partial \Gamma}+\frac{\partial \chi_1^{(1)}}{\partial r}\frac{\partial r}{\partial \Gamma}\Big)\\ &- \Big(\frac{\partial  \chi_1^{(1)}}{\partial u}\frac{\partial u}{\partial \gamma}+\frac{\partial \chi_1^{(1)}}{\partial \omega}\frac{\partial \omega}{\partial \gamma}+\frac{\partial  \chi_1^{(1)}}{\partial r}\frac{\partial r}{\partial \gamma}\Big)\Big(\frac{\partial  R_2^{(0)}}{\partial e}\frac{\partial e}{\partial \Gamma}+\frac{\partial  R_2^{(0)}}{\partial \phi}\frac{\partial \phi}{\partial \Gamma}\Big)
\end{split}
\]
(see Remark \ref{remarkCriticalTerms}). $R_2^{(0)}$ depends at most linearly on $e$. Applying the formulas of section \ref{sectionPB}, we then observe that the eccentricity is simplified in the computation of the Poisson brackets. Hence, also the terms of book-keeping order $2$ coming from $\{R_2^{(0)},\chi_1^{(1)}\}$ do not depend on $e$. This concludes the proof. 

\hfill\(\Box\)

\section*{Acknowledgements}
I.C. has been supported by the MSCA-ITN Stardust-R,
Grant Agreement n. 813644 under the H2020 research and innovation
program. C.E. also acknowledges the support of MIUR-PRIN 20178CJA2B  `New frontiers of Celestial Mechanics: theory and applications'.

\bibliographystyle{unsrt}
\bibliography{ref}

\end{document}